\documentclass[envcountsame,orivec]{llncs}
\pdfoutput=1
\usepackage[english]{babel}
\newcommand{\famille}[3]{(#1_{#2})_{#2\in#3}}
\newtheorem{fact}[theorem]{Fact}
\newcommand{\N}{\mathbb N}
\newcommand{\concept}[1]{{\bf\emph{#1}}}
\newcommand{\set}[2]{\left\{#1\ \left|\ \vphantom{#1} #2\right.\right\}}
\newcommand{\setp}[1]{\left(#1\right)}              
\newcommand{\norme}[1]{\left\|#1\right\|}

\renewcommand{\le}{\leqslant}
\renewcommand{\ge}{\geqslant}
\renewcommand{\leq}{\leqslant}
\renewcommand{\geq}{\geqslant}
\renewcommand{\phi}{\varphi}
\renewcommand{\epsilon}{\varepsilon}
\renewcommand{\setminus}{\smallsetminus}

\usepackage{comment,multicol,multirow}
\usepackage{amsfonts,amssymb,amsmath}
\usepackage[all]{xy}
\usepackage{graphicx,xcolor}
\graphicspath{{figures/}}
\usepackage{pgf,tikz,ifthen}
\usetikzlibrary{trees}
\usepackage[vlined,boxruled]{algorithm2e}
\dontprintsemicolon
\usepackage{etex,picins}

\newcommand{\probaSymbol}{\mathbb P}
\newcommand{\esperanceSymbol}{\mathbb E}
\newcommand{\Proba}[1]{\probaSymbol\left(#1\right)}
\newcommand{\Esperance}[1]{\esperanceSymbol\left(#1\right)}
\newcommand{\Psachant} [2]{\Proba{#1\left|\,\vphantom{#1} #2\right.}}     
\newcommand{\Esachant} [2]{\Esperance{#1\left|\,\vphantom{#1} #2\right.}}
\newcommand{\EsperanceFrom}[2]{{\esperanceSymbol}_{#1} \left(#2\right)}
\newcommand{\ProbaFrom}[2]{\probaSymbol_{#1} \left(#2\right)}

\newcommand{\oscil}{\triangleright}
\newcommand{\fixed}{\square}
\newcommand{\snd}{\mathop{\mathrm{snd}}}
\newcommand{\attr}{\mathop{\mathrm{attr}}}
\newcommand{\paint}{\mathop{\mathrm{paint}}}

\newcommand{\Energy}{E}
\newcommand{\Edual}{\widehat{E}}

\newcommand{\Cells}{\mathcal{V}}
\newcommand{\Edges}{\mathcal{E}}
\newcommand{\Graph}{\mathcal{G}}
\newcommand{\Trees}{\mathcal{T}}
\newcommand{\Particles}{\mathcal{P}}
\newcommand{\States}{\mathcal{Q}}
\newcommand{\Pos}{Pos}
\newcommand{\neigh}{\mathcal{N}}
\newcommand{\Att}{A}
\newcommand{\SetAtt}{\mathcal{A}}
\newcommand{\cdual}{\hat{c}}
\newcommand{\ddual}{\hat{\delta}}

\newcommand{\particule}[2]{
  \ifnum\x=#1
  \ifnum\y=#2
  edge from parent[]
  node[color=red,left=2mm,pos=0.7]{}
  \fi
  \fi
}

\begin{document}

\title{Stochastic Minority on Graphs}

\author{Jean-Baptiste Rouquier\inst{1,2} \and Damien Regnault\inst{1,2} \and \'Eric Thierry\inst{1,2,3}}
\institute{
Universit\'e de Lyon, \'Ecole Normale Sup\'erieure de Lyon, LIP, France \and
Institut des Syst\`emes Complexes Rh\^one-Alpes (IXXI), France \and
CNRS, LIAFA, Universit\'e Paris 7, France}

\date{}

\maketitle

\begin{abstract}
 Cellular automata have been mainly studied on very regular graphs carrying the
 vertices (like lines or grids) and under synchronous dynamics (all vertices update
 simultaneously). In this paper, we study how the asynchronism and the graph
 act upon the dynamics of the classical Minority rule. Minority has
 been well-studied for synchronous updates and is thus a reasonable choice to begin with.
 Yet, beyond its apparent
 simplicity, this rule yields complex behaviors when asynchronism is
 introduced.
We investigate the transitory part as well as the asymptotic behavior of the dynamics under full asynchronism (also called sequential:
 only one random vertex updates at each time step) for several types of graphs.
 Such a comparative study is a first step in understanding how the asynchronous dynamics is linked to the topology (the graph).

Previous analyses on the grid~\cite{RST-TCS2009,RST-DMTCS2010} have observed that Minority seems to induce fast stabilization.
We investigate here this property on arbitrary graphs using tools such as energy, particles and random walks.
We show that the worst case convergence time is, in fact, strongly dependent on the topology.
In particular, we observe that the case of trees is non trivial.
\end{abstract}

\section{Introduction}
In this paper, we investigate the random process Minority: each vertex of a graph is characterized by a state 0 or 1. At each time step (the time is discrete), one random subset of vertices is drawn. These vertices are updated and switch to the minority state in theirs neighborhood.

Similar random processes appeared in the literature and concern different fields of applications. For example, several studies focus on the emergence of cooperation in the iterated prisoner's dilemma game~\cite{Kit93,DGGIJ02,MR06}: used as a simple social model, this field helped determine which ingredients are needed to foster the emergence of cooperative behaviors. The rock-paper-scissors game was used to model evolution of colonies of bacteria~\cite{KRFB02}. In physics, the Ising model (a kind of Majority rule) was introduced to study ferromagnetism~(original paper is \cite{Ising25}, see for instance Chapter 10 of~\cite{stat-mechs-for-Ising} for a quick introduction): states represent the orientation of spin. Anti-ferromagnetism is studied by using a kind of Minority rule~\cite{MCT74}. Recent works~\cite{BFR09,BFR10} use stochastic Minority to model the formation of quasi-crystals.

Initially, our interest in studying Minority process comes from the field of cellular automata (CA). CA can be seen both as a model of computation with massive parallelism and as a model for complex systems in nature. They have been studied with various fields of applications like parallel/distributed computing, physics, biology or social sciences. Most of the work regarding CA assumes that their dynamics is deterministic and synchronous (all vertices update simultaneously) and that the graph is very regular (usually a line or a 2D or 3D grid). Such assumptions can be questioned with regard to the applications and the real life constraints. Dynamics where those assumptions are perturbed have been far less studied and their analysis is very challenging. Here is a non-exhaustive list of related works about CA in literature, the models are {\em stochastic CA} since the perturbations are usually introduced as stochastic processes:
\begin{itemize}
\item Perturbation of the updating rule: resilience to random errors~\cite{Gacs86,Gacs01,GR88,Toom74,Toom80,Toom90},
  mean field analysis of general Markovian rules~\cite{BBK06}.
\item Perturbation of the synchronism (i.e. of the updating scheme): empirical studies about resilience to
  asynchronism~\cite{Ingerson84,Huberman93,Ber94,Kanada94,Sch99}, mathematical
  analysis of some 1D CA under full asynchronism (only one random vertex updates
  at each time step) or under $\alpha$-asynchronism (each vertex updates
  independently with probability~$\alpha$)~\cite{FMST06,FRST06,Fuks04}.
\item Perturbation of the graph (the topology of cells): empirical studies~\cite{Fat04ACRI,FM04,RM08}, gene regulatory networks~\cite{BE93,SM03}.
\end{itemize}

Previous analyses focus on the effects of asynchronism on 2D Minority with Von Neumann neighborhood~\cite{RST-TCS2009} and Moore neighborhood~\cite{RST-DMTCS2010}. In this paper, we choose to investigate how the graph acts upon the dynamics under asynchronous updates.
(Note that this is a special case of Interacting Particle Systems~\cite{liggett2004interacting}.)
We focus on {\em Stochastic Minority} where the Minority rule applies to two possible states~$\{0,1\}$ and under full asynchronism (at each time step, only one random vertex is updated with the uniform distribution). This simple rule already exhibits a surprisingly rich behavior as observed in~\cite{RST-TCS2009,RST-DMTCS2010} where it is studied for vertices assembled into a torus. Such behaviors may appear because Minority is a CA with \emph{negative} feedback. The evolution of CA with \emph{positive} feedback can be described by a bounded decreasing function over time \cite{GM90}. Thus, the difficulty of analyzing the Minority rule (negative feedback) must not be confounded with the difficulty of analyzing the Majority rule (positive feedback).
Some related stochastic models like Ising models or Hopfield nets have been studied under asynchronous dynamics (e.g. our model of asynchronism corresponds to the limit when temperature goes to~0 in the Ising model). These models are acknowledged to be harder to analyze when it comes to arbitrary graphs~\cite{Ist00,MCT74} or negative feedbacks~\cite{Roj96}.

Let us stress that we mainly study stochastic Minority on several kinds of graphs: trees, cliques, bipartite graphs and compare them.
General results are not precise enough to describe its behavior in the present study.
One of the aim of this paper is to show that Minority's behavior highly depends on the topology of the graph.
Our paper focus on particular classes of graphs. It is a first step and
the reasonings may prove to
be helpful to study future applications of Minority (as in~\cite{BFR09,BFR10})
and our results complete some previous results about Minority
\cite{RST-TCS2009,RST-DMTCS2010}.

Here is a list of previous claims about Stochastic
Minority in the literature, as well as some insights provided by the present paper.

\noindent \textbf{Short introduction of stochastic Minority behaviors:} Figure \ref{fig:experiments} shows Minority on a $2D$ grid under three different dynamics:
\begin{itemize}
\item the \emph{$\alpha$-synchronous} dynamics where each cell has a probability $\alpha$ to be updated independently from the others at each time step;
\item the \emph{synchronous} dynamics where all cells are updated at each time step ($\alpha$-asynchronous dynamics for $\alpha=1$);
\item the \emph{fully asynchronous} dynamics where only one cell, randomly and uniformly chosen, is updated at each time step (it can be regarded as, and often behaves as the limit of the $\alpha$-asynchronous dynamics when $\alpha$ tend to $0$~\cite{FRST06,RM09}).
\end{itemize}

\begin{figure}[htbp]
\centerline{\includegraphics[width=1.08\textwidth]{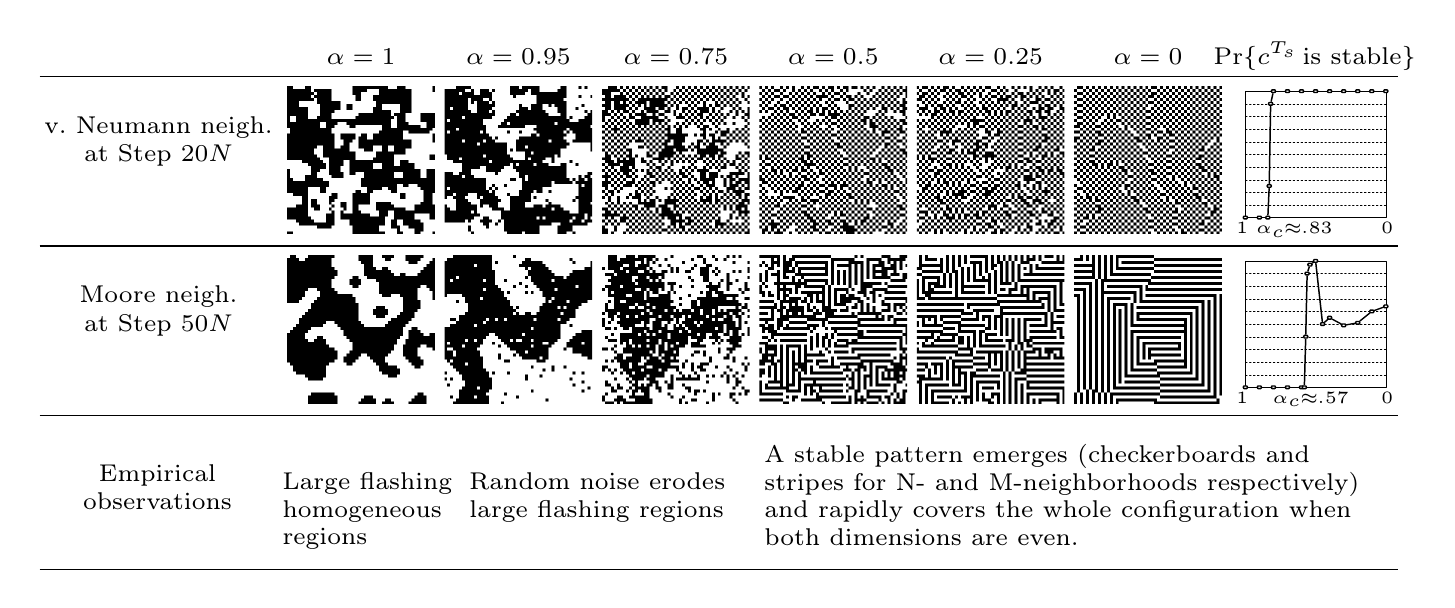}}
\caption{Stochastic Minority under different $\alpha$-asynchronous dynamics with $N_{50}=50\times50$ cells arranged in a $2D$ grid with periodic boundary condition (i.e. torus). The last column gives, for $\alpha\in[0,1]$, the empirical probability that an initial random configuration converges to a stable configuration before time step $t_s\cdot N_{50}$ where $t_s = 1000$ and $t_s = 2000$ for von Neumann and Moore neighborhood respectively. The fully asynchronous dynamics is abusively designed by $\alpha=0$.}
\vspace*{-5mm}
\label{fig:experiments}
\end{figure}

\noindent
Depending on the value $\alpha$, Minority can exhibit two different behaviors. Experimentally, these phenomenon can be observed as a phase transition depending on $\alpha$ on the convergence time (see figure \ref{fig:experiments}):
\begin{itemize}
\item When $\alpha$ is almost $1$, there is a big homogeneous \emph{flashing} background with some random noise. By flashing, we mean that all the cells of the background are not in their minority state and since $\alpha$ is almost $1$, they switch their states at almost each time step. The few cells which are not updated create some random noise. Experimentally, the dynamics never reaches a stable configuration, but a highly improbable series of updated may lead to a stable configuration.
\item When $\alpha$ is almost $0$, different regions made of checkerboard patterns (von Neumann neighborhood) or stripes (Moore neighborhood) quickly appear. The only cells which may switch their state are along the borders between these regions. These borders evolve and they eventually stabilize. Experimentally, the dynamics reaches a stable configuration in polynomial time according to the size of the configuration.
\end{itemize}
No results were previously known for the $\alpha$-asynchronous dynamics. In section \ref{sec/phase-transition}, we study this phase transition in $1D$ and show that it is linked to directed percolation. The proof of this result (theorem \ref{th/percolation}) is short because most of the technical aspects were previously done in~\cite{R08perco} on a different automaton.

\begin{figure}
\begin{tabular}{p{1.7cm} p{1.56cm} p{1.56cm} p{1.56cm} p{1.56cm} p{1.56cm} p{1.56cm} p{1.56cm}}
\small & \centering{~$t=0N$} & \centering{$t=15N$} & \centering{$t=30N$} & \centering{$t=50N$} & \centering{$t=80N$} & \centering{$t=150N$} & $t=300N$\\
\raisebox{0.85cm}{\begin{tabular}{p{1.6cm}} \tiny von Neumann neighborhood \end{tabular}} &
\includegraphics[height=1.5cm]{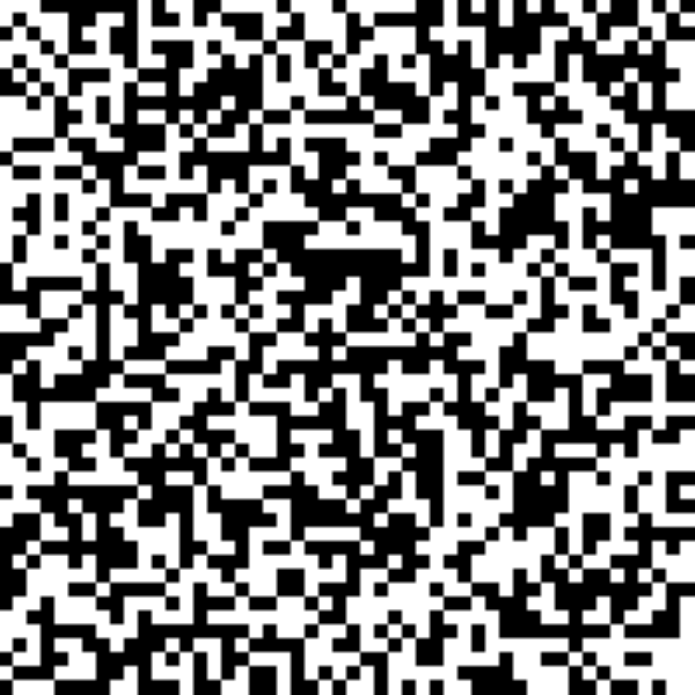} &
\includegraphics[height=1.5cm]{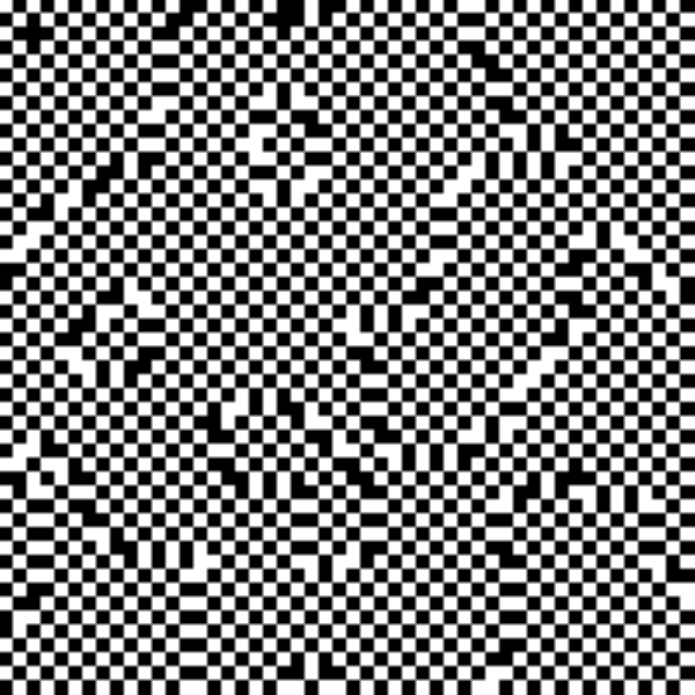} &
 \includegraphics[height=1.5cm]{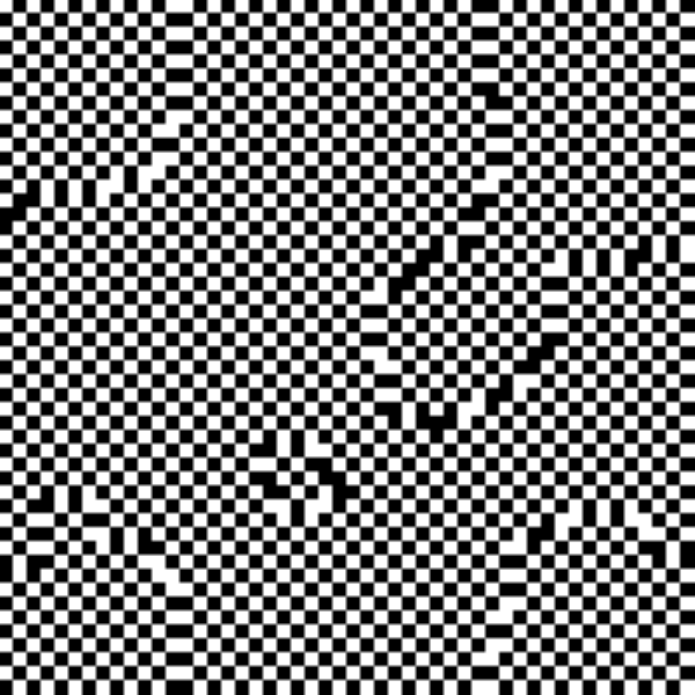} &
 \includegraphics[height=1.5cm]{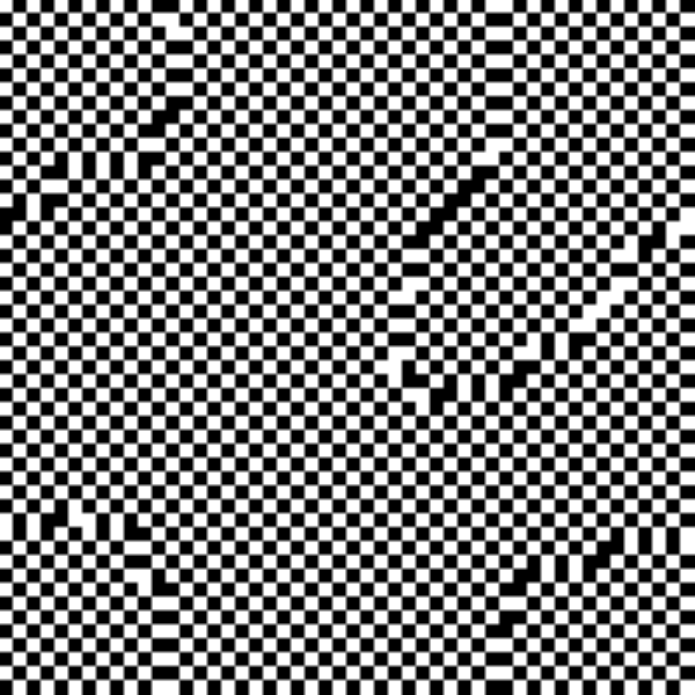} &
 \includegraphics[height=1.5cm]{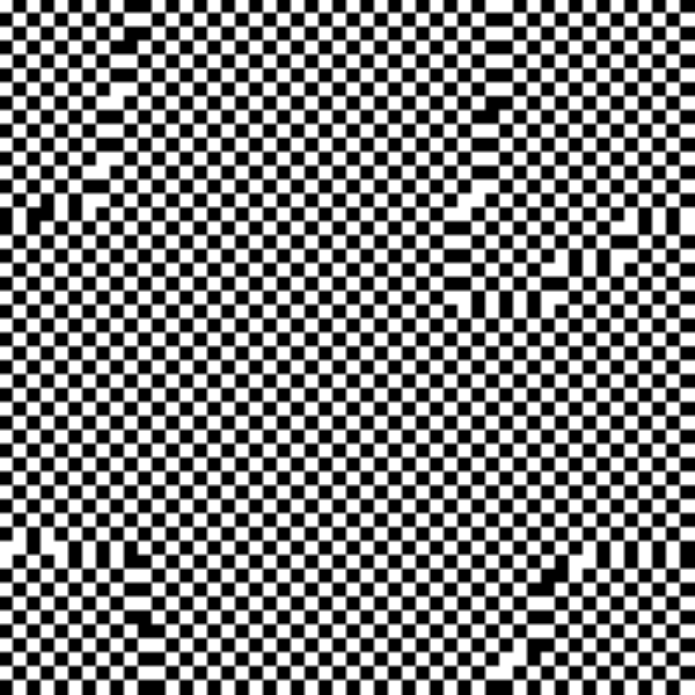} &
 \includegraphics[height=1.5cm]{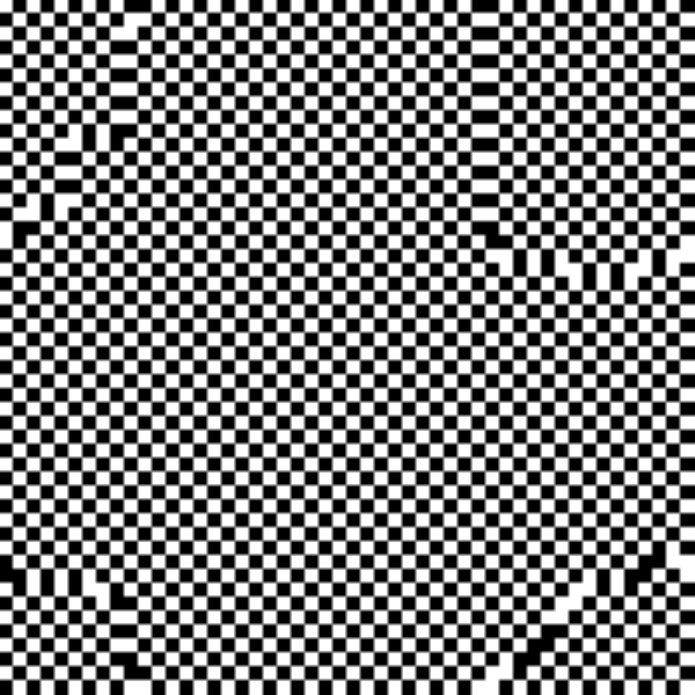} &
 \includegraphics[height=1.5cm]{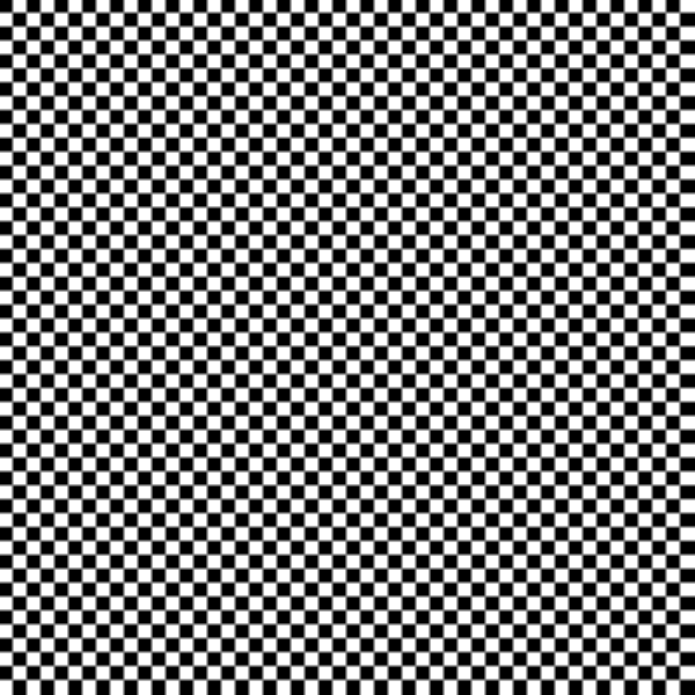} 
\\
\raisebox{0.85cm}{\begin{tabular}{p{1.5cm}} \tiny Moore neighborhood \end{tabular}} &
 \includegraphics[height=1.5cm]{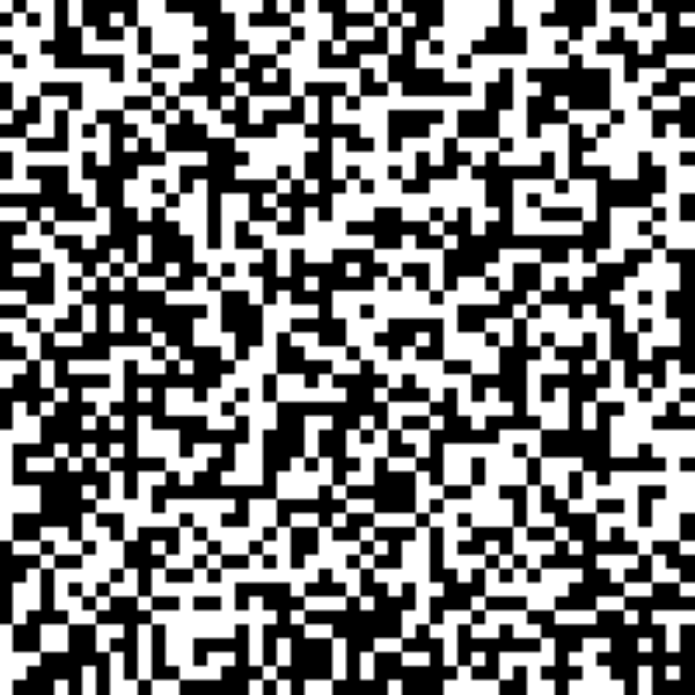} &
 \includegraphics[height=1.5cm]{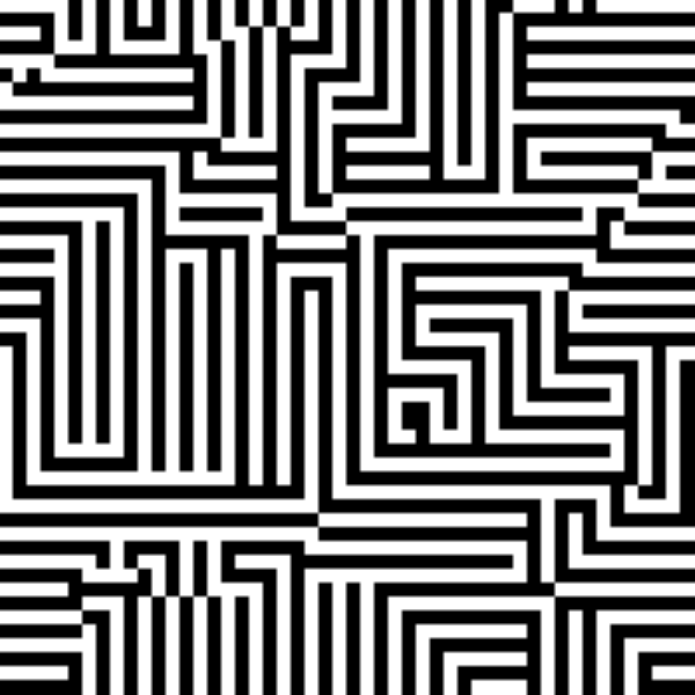} &
 \includegraphics[height=1.5cm]{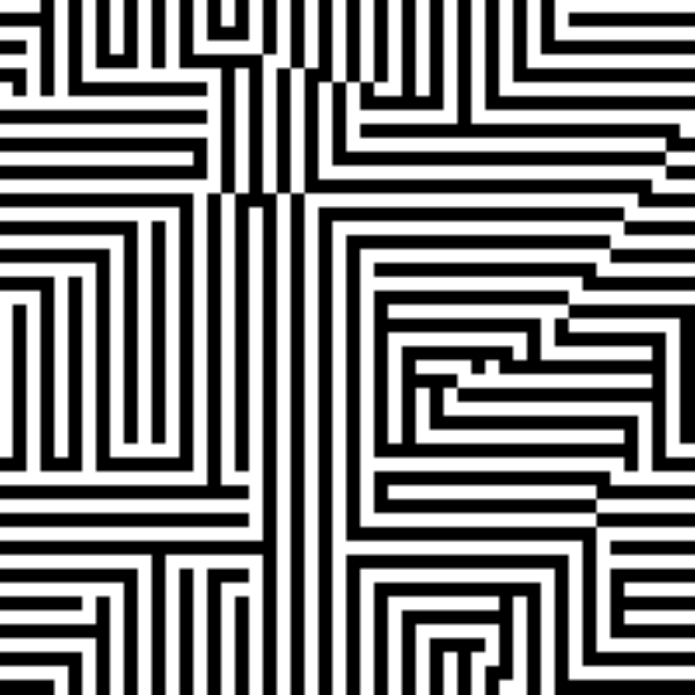} &
 \includegraphics[height=1.5cm]{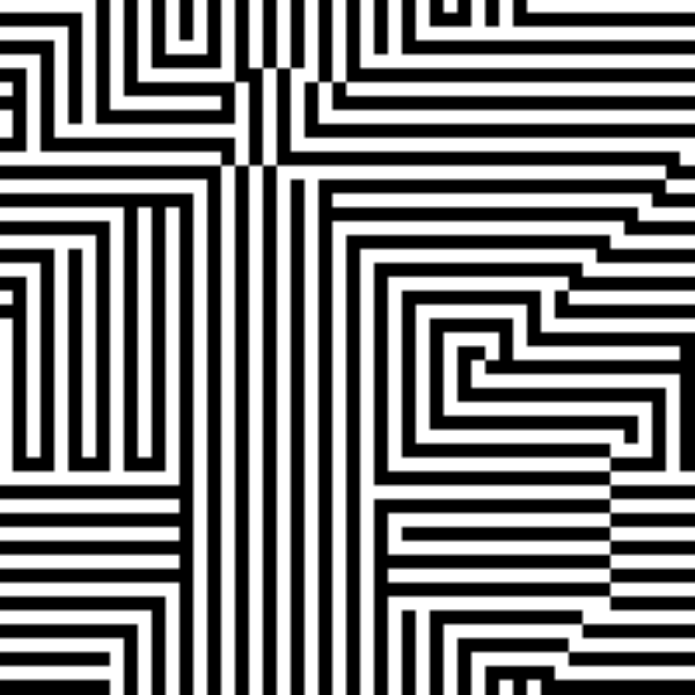} &
 \includegraphics[height=1.5cm]{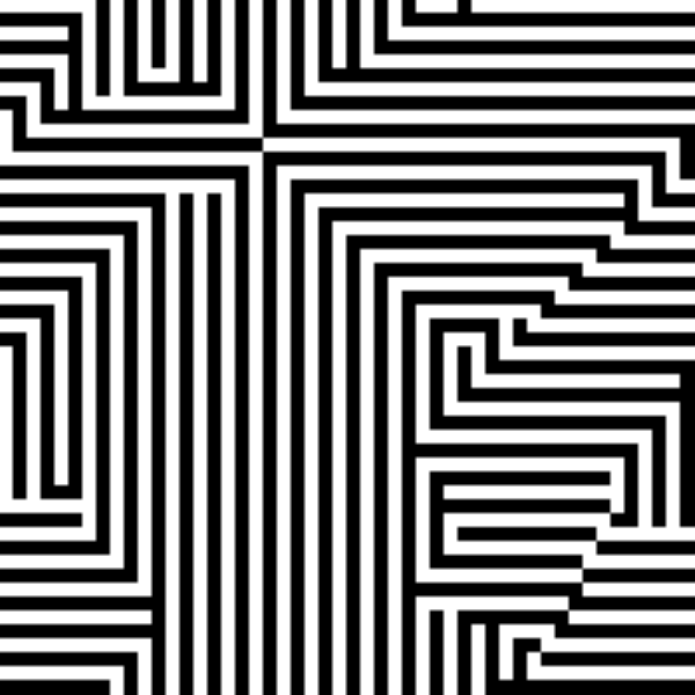} &
 \includegraphics[height=1.5cm]{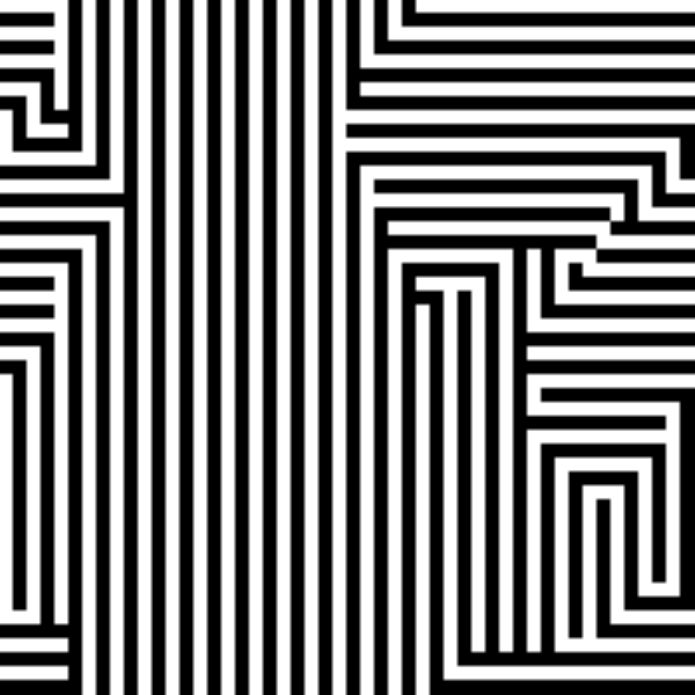} &
 \includegraphics[height=1.5cm]{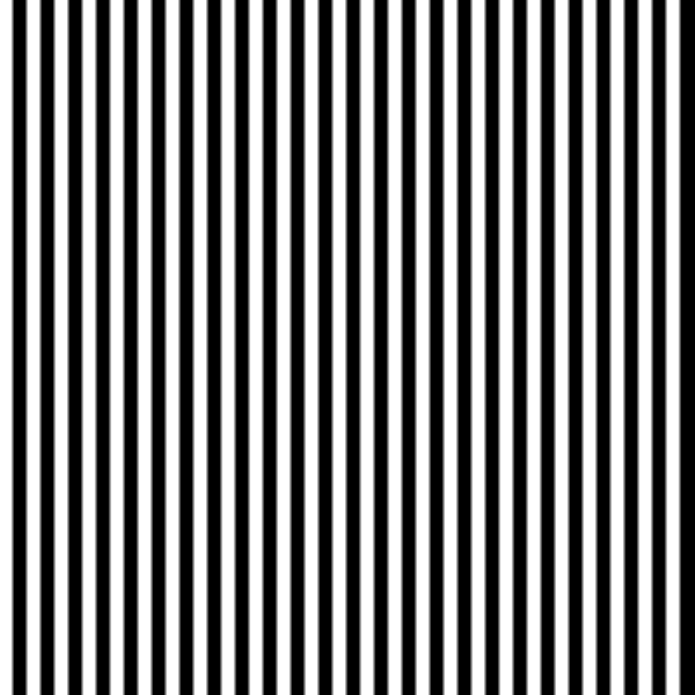} 
\\
\end{tabular}
\caption{Minority under fully asynchronous dynamics on $2D$ grids with periodic boundary condition, von Neumann and Moore neighborhood~\cite{RST-TCS2009,RST-DMTCS2010}.}
\label{vn_moore}
\end{figure}

The previous studies~\cite{RST-TCS2009,RST-DMTCS2010} focuses on the fully asynchronous dynamics which is similar to the $\alpha$-asynchronous dynamics when $\alpha$ is almost $0$ (see Figure \ref{vn_moore}). These papers analyse the start and the end of a classical execution on stochastic Minority on a $2D$ grid with von Neumann and Moore neighborhood. Indeed, they had to study separately the formation of the regions (beginning) and the evolution of the borders (ending). 

\noindent \textbf{Local interactions:} under fully asynchronous dynamics, patterns, which depend on the topology of the graph, quickly appear. In theorems $5$ and $7$ of~\cite{RST-TCS2009} and theorem $9$ of~\cite{RST-DMTCS2010}, the authors show that this phenomenon is due to \emph{local interactions} and occurs in polynomial time according to the size of the configuration. When these papers were written, the authors did not need to be more precise but they conjectured that these patterns appear in linear time. Here by studying cliques (where long range interactions do not exist), we show that the dynamics stabilizes in linear time. This result supports the previous conjecture.

\noindent \textbf{Long range interactions:} the evolution of borders between different regions often implies interactions between cells which could be arbitrarily far away in the graph. This long range interactions are harder to analyze and major differences appear between the von Neumann and Moore neighborhood on the $2D$ grid. This remark leads us to analyze different topologies. Here we solve a conjecture made in~\cite{RST-TCS2009}. We exhibit biased trees were Minority behaves as a biased random walk and need an exponential time to converge toward a stable configuration. In the previously considered topologies, Minority always converges in polynomial time under fully asynchronous dynamics.

\noindent \textbf{Capacity of simulation:} one important aspect of Minority is the diversity of phenomena embedded in this rule. On $\alpha$-asynchronous dynamics, there is a phase transition between two different behaviors. In section \ref{sec/phase-transition}, we establish a link between this phase transition and \emph{directed percolation}. Previous studies~\cite{RST-TCS2009,RST-DMTCS2010} have shown that the fully asynchronous dynamics occurs in two steps. We show here (Section~\ref{sec/decr-energy-cliq}) that the first step acts as a \emph{coupon collector} on cliques. When long range interactions occur, the dynamics may behave as competition between $2D$ regions, different kinds of \emph{random walks}. Other behaviors may be encoded on stochastic Minority. It is surprising that a simple rule may produce such different behaviors from random configurations on regular topologies. 

\noindent \textbf{Bipartite graphs:} one aim of this paper is to generalize tools previously designed to study Minority on grid. Basic tools may be generalized to any graph but we realized that advanced tools can only be generalized to bipartite graphs. This leads to an explanation between the difference of behaviors previously observed on $2D$ grids. Note that even if our tools are helpful to analyze bipartite graphs, various and complicated behaviors may already appears in this class of graphs.

\section{Model}
In this paper we consider Stochastic Minority on arbitrary undirected graphs.

\begin{definition}[Configuration]
  Let $\Graph=(\Cells,\Edges)$ be a finite undirected graph with
  vertices~$\Cells$ and edges~$\Edges$\,. $\States=\{0,1\}$ is the set of states
  ($0$ stands for white and $1$ stands for black). The vertices are also called
  cells and $N:=|\Cells|$ denotes their number. The \emph{neighborhood}
  $\neigh _i$ of a vertex $i$ is the set of its adjacent vertices (including
  itself). A \emph{configuration} is a function $c : \Cells \rightarrow \States$
  ($c_i$ denotes the state of vertex~$i$ in configuration~$c$).
\end{definition}

\begin{definition}[Stochastic Minority]
  We consider the following dynamics $\delta$ that associates with each
  configuration $c$ a random configuration $c'$ obtained as follows: a vertex
  $i\in \Cells$ is selected uniformly at random (we say that vertex~$i$ is
  \emph{fired}) and its state is updated to the minority state among its
  neighborhood (no change in case of equality), while all the other vertices remain
  in their current state. Formally:
  $$(\delta(c))_i = \left\{ \begin{array}{cl}
    1 & \textup{if $ \sum_{j \in \neigh_i} c_j < \frac{|\neigh_i|}{2}$
      \quad or\quad $\sum_{j \in \neigh_i} c_j = \frac{|\neigh_i|}{2}$ and $c_i = 1$}\\[1mm]
    0 & \textup{if $\sum_{j \in \neigh_i} c_j > \frac{|\neigh_i|}{2}$
      \quad or\quad $\sum_{j \in \neigh_i} c_j = \frac{|\neigh_i|}{2}$ and $c_i = 0$}
  \end{array}\right.$$
and $(\delta(c))_k = c_k$ for all $k \neq i$.
In a configuration, a vertex is said to be \emph{active} if its state changes when the vertex is fired.
The random variable $c^t$ denotes the configuration obtained from an initial
configuration $c$, after $t$ steps of the dynamics: $c^0=c$ and $c^t =\delta^t(c)$ for all $t \geq 1$.
(The notation $\delta^t(c)$ means $\delta\circ...\circ\delta(c)$, $t$ times).
\end{definition}

\begin{definition}[Attractors and limit set]
  For the dynamics induced by~$\delta$, a set of configurations $\Att$ is an
  {\em attractor} if for all $c,c' \in \Att$, the time to reach~$c'$ starting
  from~$c$ is finite almost surely.
  In the transition graph where vertices are all the possible configurations and arcs $(c,c')$ satisfy $\Proba{\delta(c)=c'}>0$,
  attractors are the strongly connected components with no arc leaving the component.
  The union of all attractors is denoted $\SetAtt$ and called the {\em limit set}.
\end{definition}

\begin{definition}[Convergence and hitting time]
  We say that the \emph{dynamics $\delta$ converges} from an initial
  configuration $c^0$ to an attractor $\Att$ (resp. the limit set $\SetAtt$) if
  the random variable $T=\min \set{t}{c^t \in \Att}$
  (resp. $T=\min \set{t}{c^t \in \SetAtt}$) is almost surely finite.

The variable~$T$ is a \emph{hitting time}.
\end{definition}
Since we only consider finite graphs, the dynamics $\delta$ always converges from any initial configuration to $\SetAtt$,
and $T$ is well-defined
(an exception is section~\ref{sec/phase-transition}, 
where we consider an infinite graph and discuss the convergence to an attractor).

\bigskip
A hitting time $T$ is defined for a given graph and a given initial configuration.
We are interested in the worst case (i.e. largest hitting time)
among all possible initial configurations and among all graphs of a given class of graphs.

\section{Tools}

\subsection{Energy, Potential}
As in the Ising model~\cite{MCT74} or in Hopfield networks~\cite{Roj96}, we
define a natural global parameter similar to an energy:
it counts the number of interactions between neighboring vertices in
the same state. This parameter will provide key insights into the evolution of the
system.

\begin{definition}[Potential]
The \emph{potential $v_i$} of vertex $i$ is the number of its neighbors (including itself) in the same state as itself. 
If $v_i \leq \frac{|\neigh_i|}{2}$ then the vertex is in the minority state and is
thus inactive; whereas, if $v_i > \frac{|\neigh_i|}{2}$ then the vertex is
active. A configuration $c$ is \emph{stable} if and only if for all vertex~$i\in
\Cells$, $v_i \leq \frac{|\neigh_i|}{2}$.
\end{definition}

\begin{definition}[Energy]
The \emph{energy} of configuration~$c$ is $\Energy (c) = \sum_{i \in \Cells} (v_i-1)$.
\end{definition}

The energy of a configuration is always non-negative.
There are configurations of energy~0 if and only if $\Graph{}$ is bipartite:
those stable configurations are the 2-colorings of~$\Graph$. More generally, we have:

\begin{proposition}[Energy bounds]
\label{prop:bounds:energy}
The energy $\Energy$ satisfies $2|\Edges|-2 C_{\max} \leq \Energy \leq 2|\Edges|$,
where $C_{\max}$ is the maximum number of edges in a cut of~$\Graph$.
\end{proposition}

\begin{proof}
The bounds are direct consequences of the definitions:
for any configuration, $\Energy=2|\Edges|-2|C|$ where $C$ is the cut $\set{\{i,j\} \in \Edges}{c_i=0\text{ and }c_j=1}$.
\end{proof}

As a consequence, computing the minimum energy for arbitrary graphs is NP-hard: it is equivalent to computing \textsc{maxcut}.

\begin{lemma}[Energy is non-increasing]
The energy is a non-increasing function of time and decreases each time a vertex $i$ with potential
stricly larger than $\frac{|\neigh_i|}{2}$ fires.
\end{lemma}
\begin{proof}
When an active vertex $i$ of potential $v_i$ fires, its potential becomes $|\neigh_i|-v_i+1$,
and the energy of the configuration becomes $\Energy + 2|\neigh_i|-4v_i+2$.
\qed
\end{proof}

\begin{corollary}
\label{prop:limit-set-iff-energy-cannot-decrease}
A configuration $c$ belongs to the limit set if and only if no sequence of updates would lead the energy to decrease,
i.e. if and only if $\forall t$, $\Psachant{\Energy(c^t)>\Energy(\delta(c^t))}{c^0=c}=0$.
\end{corollary}
\begin{proof}
If the energy decreases when updating $c$ to $c'$, then $c$ will never be reached again (because energy is non-increasing).
Reciprocally, any update that keeps the energy constant is reversible:
the fired vertex can be fired again to get back to the previous configuration.
\qed
\end{proof}

\begin{remark}
\label{rem:odd-degree}
Since firing a vertex of odd degree makes the energy decrease, such vertices are inactive in the limit set.
\end{remark}

\begin{definition}[Particle]
  Let $c$ be a configuration on~$\Graph=(\Cells,\Edges)$, an edge $\{i,j\}$
  holds a {\em particle} if $c_i=c_j$. A configuration is fully characterized
  (up to the black/white symmetry) by its set of particles located at $\Particles \subseteq \Edges$.
\end{definition}
(Note that the converse proposition ``any subset $\Particles \subseteq \Edges$ corresponds to a configuration'' is true if and only if the graph is a tree).

The energy of a configuration is clearly equal to twice its number of particles. 
With the particle point of view, when firing a vertex~$i$ of
degree~$deg(i)=|\neigh_i|-1$, if the number of incident edges holding a particle
is a least $\frac{deg(i)}{2}$, these particles disappear but new particles appear
on the incident edges (if any) which had no particle (as illustrated on
Fig.~\ref{fig:particles1D}). Otherwise the particles do not move.

\begin{figure}[htb]
\centering
\includegraphics[height=1.2cm]{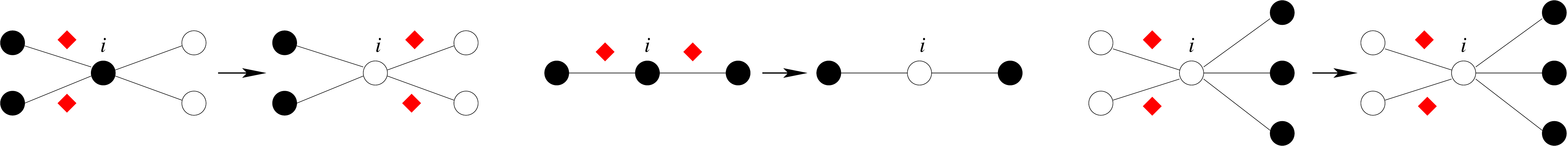}
\caption{Transfers of particles (diamonds) when firing vertex~$i$.}
\label{fig:particles1D}
\end{figure}

Switching between the coloring and the particle points of view may simplify the
description of the configurations and the dynamics, e.g. when the energy is low
and the dynamics comes to random walks of a few particles.

\subsection{Bipartite Graphs}
A graph is bipartite when its vertices can be partitioned into two sets such that every edge goes from one set to the other
(or equivalently, when it is 2-colorable).
Bipartite graphs allows us to use another, easier point of view (the dual configuration)
and to easily determine if a configuration is in the limit set.

\subsubsection{Dual configuration}
We now introduce dual configurations as in~\cite{RST-TCS2009} (section 3.2),
and their dual rule to facilitate the study of
the dynamics on trees. In this dual dynamics equivalent to Stochastic Minority,
the stable configurations of minimum energy are the two configurations all black
and all white, and the regions which compete are all white versus all black
subtrees.

The dual rule is almost a majority rule,
but in case of equality among the neighbors of a vertex, the state of this vertex is flipped each time it is updated.
This ``instability'' prevents many results about majority rules to apply to our case.

\begin{definition}[Dual configurations]
\label{def:dual_conf}
Consider a graph $\Graph$ and fix a vertex $r$ (the ``root'').
For any configuration $c$ on~$\Graph$, its dual configuration $\cdual$ is defined as
$\cdual_i = c_i$ if $h_i$ is even and $\cdual_i = 1-c_i$ if $h_i$ is odd, where $h_i$ is the
distance from~$r$ to~$i$ (see Fig.~\ref{fig:dual}). The mapping $c \mapsto
\hat{c}$ is a bijection on the set of all the configurations; more precisely
$\hat{\hat{c}}=c$.
\end{definition}

An equivalent definition consists in making a XOR with the 2-coloring
of~$\Graph$ such that $r$ is white.
The duals of the configurations of minimum energy~0 are the configuration
where all vertices are black or all vertices are white.

\begin{figure}[htb]
\centering
\includegraphics[height=1.6cm]{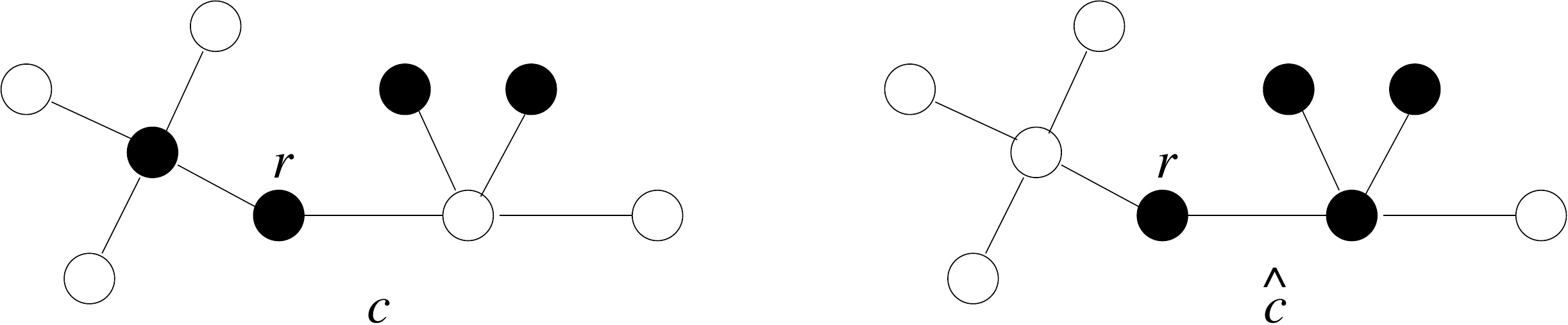}
\caption{A configuration~$c$ and its dual configuration~$\hat{c}$ (with regard to root~$r$).}
\label{fig:dual}
\end{figure}

\begin{proposition}[Dual dynamics]
  Consider a sequence $(c^t)$ for the Stochastic Minority dynamics~$\delta$ and
  the sequence $(\cdual^t)$ of the dual configurations, and define the {\em dual
    dynamics} $\ddual$ as $\ddual(\cdual)=\widehat{\delta(c)}$ so that
  $\cdual^{t+1}=\ddual(\cdual^t)$. Then the dynamics $\ddual$ is also a
  stochastic CA. It associates with each configuration~$\cdual$ a random
  configuration~$\cdual'$ by updating one random vertex~$i$ uniformly with the
  rule which selects the majority state in the neighborhood of~$i$ excluding
  itself (in case of equality its state changes):
$$
\cdual'_i =
\left\{
  \begin{array}{cl}
    1 & \textup{if $\sum_{j \in \neigh_i\setminus \{i\} } \cdual_j > \frac{|\neigh_i|-1}{2}$
               or ($\sum_{j \in \neigh_i\setminus \{i\} } \cdual_j = \frac{|\neigh_i|-1}{2}$
               and $\cdual_i = 0$)}\\[1mm]
    0 & \textup{if $\sum_{j \in \neigh_i\setminus \{i\}} \cdual_j < \frac{|\neigh_i|-1}{2}$
               or ($\sum_{j \in \neigh_i \setminus \{i\}} \cdual_j = \frac{|\neigh_i|-1}{2}$
               and $\cdual_i = 1$)}
\end{array}\right.
$$ 
By construction, the dual sequences $(c^t)$ and $(\cdual^t)$ as well as their
corresponding dynamics $\delta$ and $\ddual$ are {\em stochastically coupled}
(see~\cite{Lind92}) by firing the {\em same} random vertex at each time step.
\end{proposition}

\begin{definition}[Dual potential \& Energy]
  The \emph{dual potential $\hat{v}_i$} of vertex $i$ is the number of its
  neighbors (excluding itself) in a different state than itself. If~$\hat{v}_i <
  \frac{|\neigh_i|-1}{2}$ then the vertex is in the majority state and is thus
  inactive; whereas, if $\hat{v}_i \geq \frac{|\neigh_i|-1}{2}$ then the vertex is
  active. The {\em dual energy} $\Edual$ is the sum of the dual potentials over
  all the vertices.
\end{definition}

Given a configuration $c$ and its dual $\cdual$, the potential of any vertex~$i$
in~$c$ is equal to the dual potential of vertex~$i$ in~$\cdual$ plus 1. Thus the
dual energy of~$\cdual$ is exactly the energy of~$c$.

\begin{theorem}
Consider a configuration $c$ and its dual $\cdual$. Then, if the same vertex fires, $\widehat{\delta(c)} = \ddual(\cdual)$.
We thus have the following commutative diagram:
\begin{center}
\xymatrix{
\txt{original} & c  \ar@{<->}[d] \ar[r]^{\displaystyle\delta} & c' \ar@{<->}[d] \\
\txt{dual}     & \cdual          \ar[r]^{\displaystyle\ddual} & \cdual'
}
\end{center}
\end{theorem}

\begin{proof}
 Consider a configuration $c$ and its dual $\cdual$. A vertex $c_i$ of $c$ is
 active if and only if $v_i \leq \frac{\neigh_i}{2}$. Since $v_i=\hat{v}+1$
 and $\hat{v}_i$ is an integer, a vertex $c_i$ of $c$ is active if and only if
 $\hat{v}_i < \frac{\neigh_i-1}{2}$: that is to say vertex $\cdual_i$ of
 $\cdual$ is active. According to definition~\ref{def:dual_conf}, if the same
 vertex fires and a vertex of $c$ is active if and only if the corresponding vertex
 of $\cdual$ is active then $\widehat{\delta(c)} = \ddual(\cdual)$.
\qed
\end{proof}

\subsubsection{Distance to a Stable Configuration}
In this section we describe an algorithm (algorithm~\ref{algo:limit-set}) that
gives a sequence of updates that leads to the limit set.
It is then easy to test for the limit set:
the input configuration belongs to the limit set if and only if the energy is the same between the input and output configurations.

\bigskip
\begin{fact}
\label{prop:attractor-decomposition}
  An attractor $A$ decomposes the graph into three sets of vertices:
\begin{enumerate}
\item\label{item:always0} the vertices that are in the state $0$ for every configuration of $A$;
\item\label{item:always1} the vertices that are in the state $1$ for every configuration of $A$;
\item\label{item:oscillating} the vertices that can be either in the state $0$ or $1$, depending on the configuration in $A$.
\end{enumerate}
\end{fact}

\begin{algorithm}[H]
\label{algo:limit-set}
\KwIn{A configuration $c$.}
\lnl{noircir}\lWhile{There is an active white vertex $i$}{Fire $i$ ($i$ becomes black)}\;
\lnl{blanchir}\lWhile{There is an active black vertex $i$}{Fire $i$}\;
\lnl{renoircir}\lWhile{There is an active white vertex $i$}{Fire $i$}\;
\KwOut{The configuration $c'$ at the end of phase~\ref{blanchir}.}
\caption{Membership to the limit set: check that $\Edual(c')=\Edual(c)$.}
\end{algorithm}

\begin{proposition}
\label{prop:algorithm-limit-set}
  The configuration $c'$ returned by algorithm~\ref{algo:limit-set} is in the limit set.
\end{proposition}
\begin{corollary}
  The input configuration $c$ is in the limit set
  if and only if the energy has not decreased during execution of the algorithm.
\end{corollary}
\begin{proof}
  We first prove that the vertices in state $0$ (white) at the end of phase~\ref{noircir} of the algorithm
  cannot switch to state $1$, whatever the sequence of updates, i.e.
  they are in Case~\ref{item:always0} of Fact~\ref{prop:attractor-decomposition}.
  Indeed, assume instead that there exists a vertex $i$ and sequence of configurations $c^1, c^2, ..., c^k$ such that
  \begin{itemize}
  \item $c^1$ is the configuration at the end of phase~\ref{noircir};
  \item each configuration is the result of firing one vertex in the previous configuration;
  \item $c^1_i=0$ and $c^k_i=1$.
  \end{itemize}
  Let $c^\ell$ be the configuration just before the first update of this sequence that fires an active vertex $j$ in state $0$
  (there exists one since at least $i$ will be fired in this sequence).
  But $j$ must have been already active in $c^1$, since it had at least as many neighbors in state $1$ as in $c^\ell$
  (this is a monotonicity argument).
  The algorithm thus would not have exited the first ``while'' loop, which is a contradiction.

  Same holds for the vertices in state $0$ at the end of phase~\ref{renoircir}.
  Also, by symmetry, the vertices in state $1$ at the end of phase~\ref{blanchir}
  are in Case~\ref{item:always1} of Fact~\ref{prop:attractor-decomposition}.

  Finally, since the remaining vertices
  were white at the end of phase~\ref{blanchir},
  they can be made white by a sequence of updates starting from the configuration $c^1$
  (the configuration at the end of phase~\ref{noircir}).
  By monotonicity, they can be made white by the same sequence of updates, starting from any configuration reachable from $c^1$.
  This means that the configuration where all those remaining vertices are white is always reachable,
  and is thus in the attractor.
  This configuration is the one at the end of phase~\ref{blanchir}. Which concludes the proof.

  By symmetry, those remaining vertices can be made black from the configuration at the end of phase~\ref{blanchir},
  which is in the attractor, so they are in Case~\ref{item:oscillating} of Fact~\ref{prop:attractor-decomposition}.
\qed
\end{proof}

Termination of Algorithm~\ref{algo:limit-set} is clear: phase~\ref{noircir}
increases the number of vertices in state $1$ at each iteration and is thus
completed in $O(N)$ iterations. The same remark applies to the other phases.

Three phases are necessary to correctly classify the vertices into the three cases of fact~\ref{prop:attractor-decomposition},
as shown on Fig.~\ref{fig:3phases}.
Two phases are enough to reach the limit set.
\begin{figure}[htb]
  \centering
  \begin{minipage}[c]{0.34\linewidth}
    \scalebox{0.8}{  \begin{tikzpicture}[grow cyclic,shape=circle]
    \tikzstyle{level 1}=[level distance=10mm,sibling angle=90]
    \tikzstyle{level 3}=[level distance=10mm,sibling angle=60]
    \tikzstyle{every node}=[,draw=black,shape=circle]
    \coordinate[]
    node[fill=white]{$a$}
    child{node{}}
    child foreach \x in {1,2} {
      node[fill=white]{\phantom{0}}
      child foreach \x in {1,2} {node[fill=white]{\phantom{0}}}
    }
    child {
      node[fill=black,text=white]{\phantom{1}}
      child foreach \x in {1,2} {node[fill=black,text=white]{\phantom{1}}}
    }
    child {
      node[fill=black,text=white]{$b$}
      child foreach \x in {1,2} {
        node[fill=white]{\phantom{0}}
        child foreach \x in {1,2} {node[fill=white]{\phantom{0}}}
      }
    };
  \end{tikzpicture}}
  \end{minipage}
  \begin{minipage}[c]{0.64\linewidth}
    Starting from this configuration, 
    phase~\ref{noircir} of the algorithm makes the vertex $a$ become black.
    Then, phase~\ref{blanchir} makes its left neighbor $b$ become white, as well as $a$.

    $b$ is now definitely white (Case~\ref{item:always0} of Fact~\ref{prop:attractor-decomposition}),
    but was not identified as such at the end of phase~\ref{noircir}.
    Thus, a third phase is necessary with this algorithm to identify $b$.

    \medskip
    No other vertex will be active during the execution of the algorithm.
  \end{minipage}
  \caption{Three phases are necessary.}
  \label{fig:3phases}
\end{figure}

\subsubsection{Stable Configurations}
Stable configurations on trees are characterized in the appendix, Section~\ref{sec/char-stable-conf}.

\subsection{Non Bipartite Graphs}
It seems that the only in-depth study of a non-bipartite graph is the analysis of
the $2D$ grid with Moore neighborhood and periodic boundary condition~\cite{RST-DMTCS2010}. This study is harder
and more complicated than any bipartite graphs considered so far.
This complexity appears in the structure of the stable configurations. For
example, Figure \ref{fig:stab} shows a stable configuration, illustrating many ways in which the
four different stable patterns can be intricated.
Perturbing such stable configurations leads to original random walks
(see Figure \ref{fig:particles2D}). Particles appear and move along the
borders. They deform the border when they move. When two particles collide, they
disappear. When borders are too perturbed, the whole structure of the stable
configuration collapses.
As opposed to the bipartite case, we are currently not able to analyze such dynamics and 
compute for any configuration a sequence of updates leading to a stable
configuration.

\begin{figure}[htb]
\centerline{\includegraphics[scale=0.6]{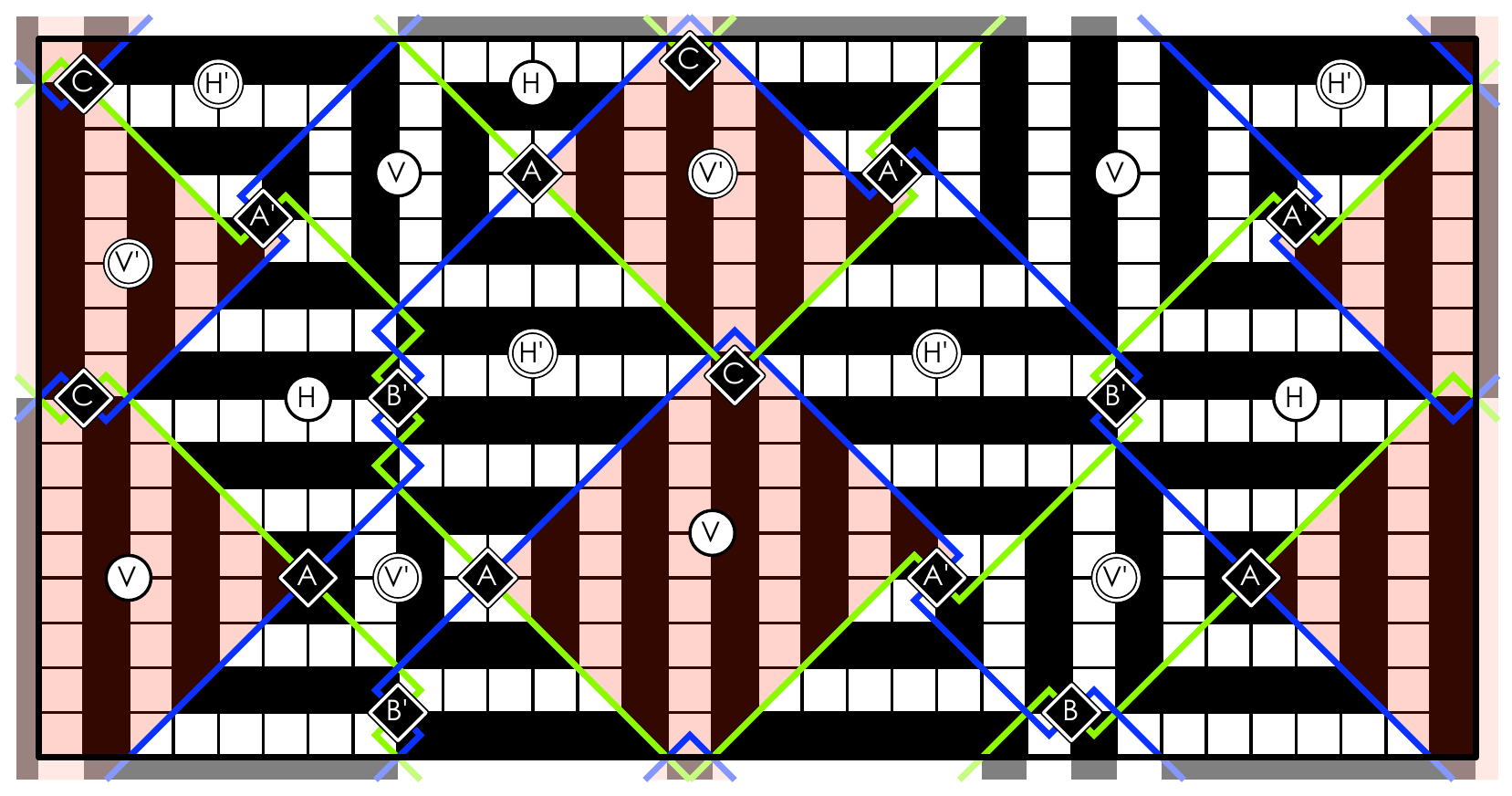}}
\caption{A complicated $16\times32$ stable configuration. There are four stable patterns (V,V',H,H') and different types of junction between the different regions(A, A', B, B' and C).}
\label{fig:stab}
\end{figure}

\begin{figure}[htb]
\centerline{\scriptsize
\begin{tabular}{p{7cm}p{6.5cm}}
    \,\! \hfill \includegraphics[height=2.2cm]{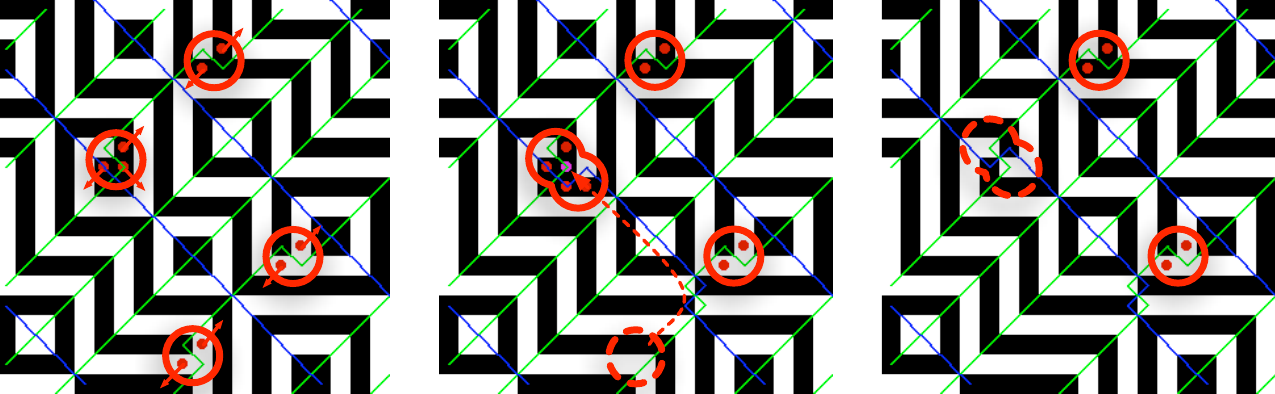} \hfill ~
&   \,\! \hfill  {\includegraphics[height=2.2cm]{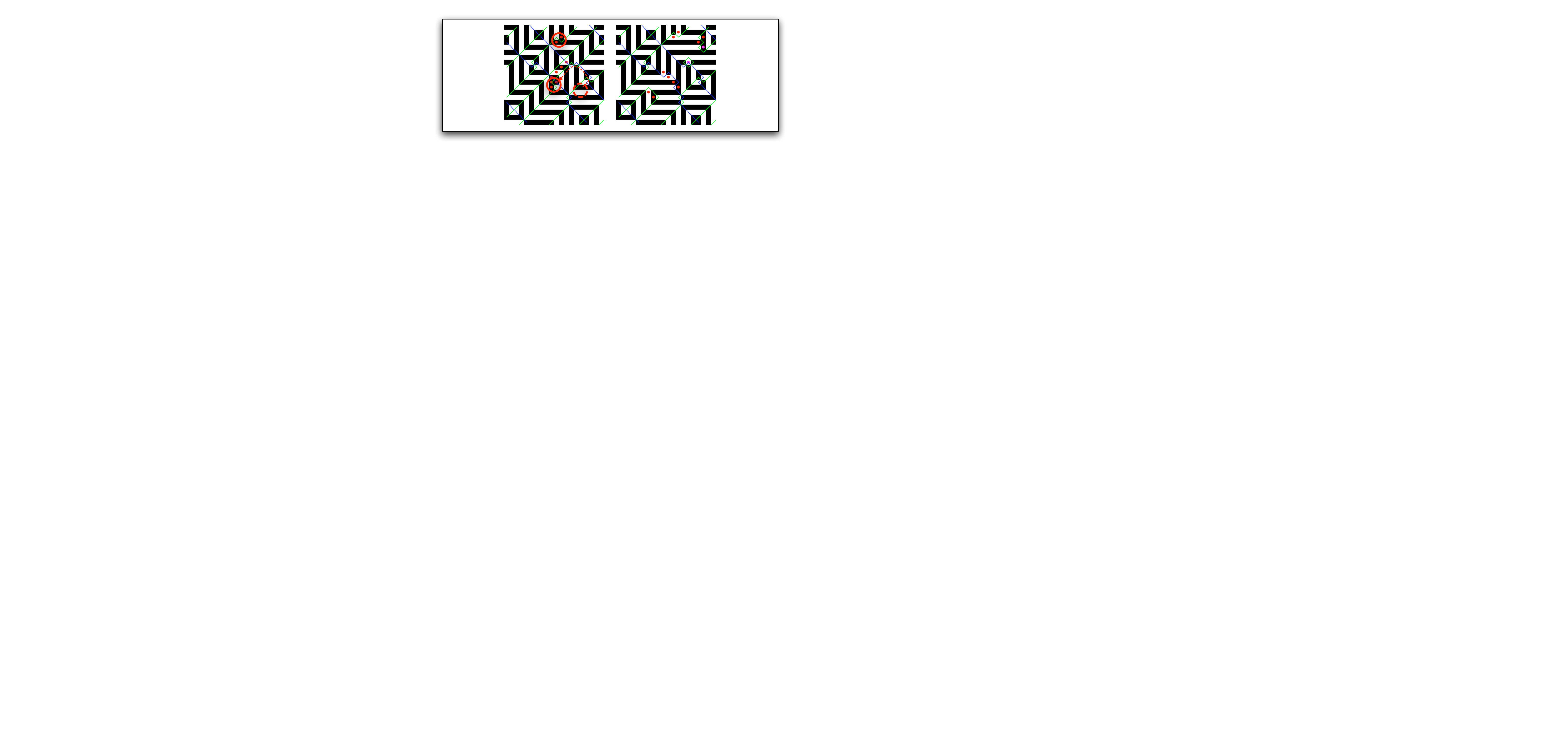}} \hfill~ \\[2mm]
\small \ref{fig:particles2D}.a -- A sequence of updates in a configuration starting with 4 particles where two of them move along rails and ultimately vanish after colliding with each other.
&
\small \ref{fig:particles2D}.b -- A sequence of updates where the rails cannot sustain the perturbations due to the movements of the particles: at some point, rails get to close with each other, new active cells appear, and part of the rail network collapses.
\end{tabular}}
\caption{Some examples of the complex behavior of particles in a 20 $\times$ 20 configuration.}
\label{fig:particles2D}
\end{figure}

\section{Behaviors}

\subsection{Decreasing Energy (Clique, cycle and paths)}
\label{sec/decr-energy-cliq}
\noindent
\begin{minipage}{0.84\linewidth}
We prove in this part that Stochastic Minority on cliques behaves as a {\em coupon collector} (see~\cite{GrimmettStirzaker2001}, page 210).
This easy result implies a fast convergence to the limit~set.
\end{minipage}%
\begin{minipage}{0.15\linewidth}
\begin{center}
\includegraphics[height=1.6cm]{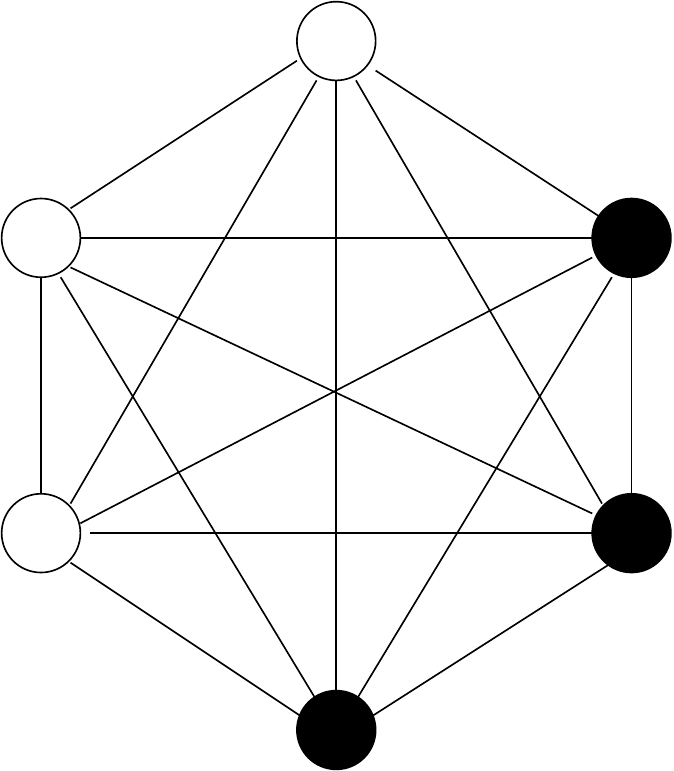}
\end{center}
\end{minipage}

\begin{theorem}
Stochastic Minority on cliques hits the limit set after $O(N\log N)$ steps on expectation.
If $N$ is even, the $\binom{N\rule[-0.5ex]{0pt}{0pt}}{\frac{N}{2}}$ attractors
are the  half black and half white configurations, each one stable.
If $N$ is odd, the only attractor is the set of $2\binom{N\rule[-0.5ex]{0pt}{0pt}}{\frac{N-1}{2}}$ configurations
having one more (resp less) black than white vertices.
\end{theorem}
\begin{proof}
  Let $n_b$ be the number of black vertices of a configuration. Since the
  neighborhood of a vertex is~$\Cells$, the potential of a black vertex
  is~$n_b$.
  If $n_b > \frac{N+1}{2}$
  (resp.~$n_b < \frac{N-1}{2}$) then firing a black (resp. white) vertex decreases
  the energy and the configuration is not in $\SetAtt$. 
  Thus a configuration in
  $\SetAtt$ must have $\frac{N-1}{2} \leq n_b \leq \frac{N+1}{2}$. Consider such a
  configuration:
\begin{itemize}
\item If $N$ is even then all vertices have potential $\frac{N}{2}$ and these
  $\binom{N}{\frac{N}{2}}$
  configurations are stable.
\item If $N$ is odd we call $C_b$ (resp. $C_w$) the set of configurations where
  $n_b = \frac{N+1}{2}$ (resp.~$n_b = \frac{N-1}{2}$). White (resp. black) vertices
  of a configuration in $C_b$ (resp.~$C_w$) are inactive and black (resp. white)
  vertices are active, firing one of them leads with constant energy to a configuration of~$C_w$
  (resp.~$C_b$). Thus from any configuration of~$C_w \cup C_b$, there is no
  sequence of updates that causes a drop of energy and~$C_w \cup C_b
  =\SetAtt$.

  Now, we prove that $\SetAtt$ is made of only one attractor.
  Let $d(c,c'):= |\set{i}{c_i\neq c'_i}|$ be the distance between two configurations.
  Consider in $\SetAtt$ two configurations $c\neq c'$.
  By symmetry we can assume $c\in C_b$.
  Since $c$ has at least as many black vertices as $c'$,
  there is a vertex $i$ black in $c$ and white in $c'$.
  Firing $i$ decreases the distance.
  Iterating this argument, one finds a path from $c$ to $c'$ in $\SetAtt$.
\end{itemize}

Now consider a configuration where $n_b > \frac{N+1}{2}$. As long as the
configuration does not belong to $\SetAtt$, the white vertices are
inactive. When $\left\lceil n_ b - \frac{N+1}{2} \right\rceil$ black vertices have fired, the
configuration is in $\SetAtt$. At each time step there is a probability
$\frac{n_b}{N}$ to fire a black vertex. This kind of dynamics is known as coupon
collector and $T = O(N\log N)$.
\qed
\end{proof}

Moreover, it is easy to see that in the odd $N$ case, the attractor has a simple structure:
this is a bipartite graph, composed of configurations with $\frac{N-1}2$ white vertices on one side,
and configurations with $\frac{N+1}2$ white vertices on the other.

\subsubsection{Cycles and paths}
Cycles and paths forms the class of connected graphs of maximum degree 2.

\label{subsec:cycle}
\mbox{}\par\noindent
\begin{minipage}{0.84\linewidth}
  On cycles and paths, the particle point of view is convenient and one can prove that
  Stochastic Minority behaves as {\em random walks of annihilating particles} on
  a discrete ring (see~\cite{FMST06,GP06}). On the right, the particles are the diamonds.
\end{minipage}%
\begin{minipage}{0.15\linewidth}
\begin{center}
\includegraphics[height=1.6cm]{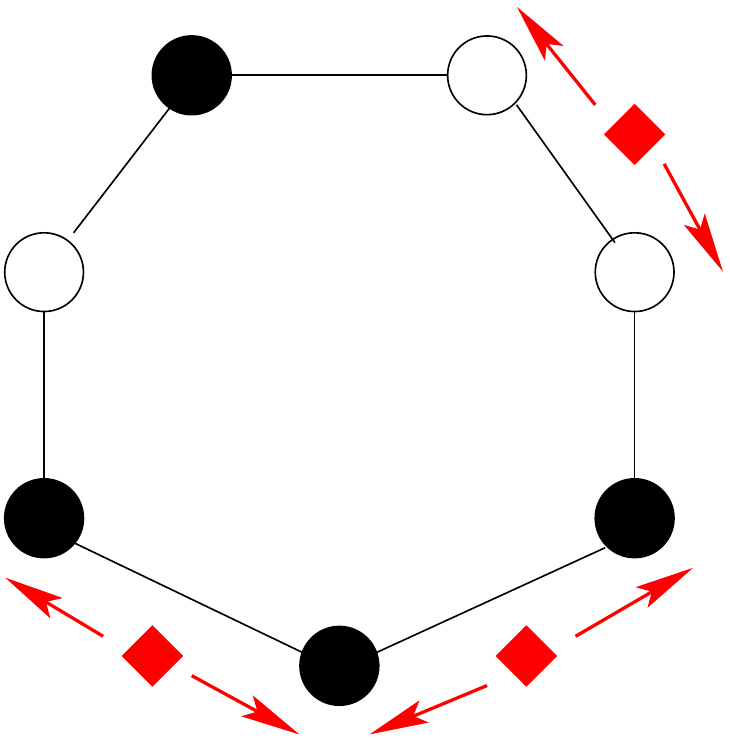}
\end{center}
\end{minipage}

\begin{theorem}
\label{prop:cycle-is-random-walk}
Stochastic Minority on cycles and paths hits the limit set after $O(N^3)$ steps on expectation.
If $N$ is even, the two attractors are the 2-colorings of the cycle. 
If $N$ is odd, the single attractor is a cycle in the transition graph composed of all the configurations with only one particle.
\end{theorem}
\begin{proof}
  The movement of particles provide a nice framework for the analysis. Firing a
  vertex incident to a particle either attract the particle on the next edge if it
  is free, or annihilate both particles on each side of the vertex. Consequently,
  the dynamics boils down to the analysis of random walks of annihilating particles
  on a discrete ring. If the number $N$ of vertices is even (resp. odd), any
  configuration has necessarily an even (resp. odd) number of particles.
  This number decreases by annihilations until there is no (resp. only one) particle.
  The attractor is reached at this point since the energy cannot decrease further.
  This proof can be easily adapted to graphs which are paths.

\medskip
  To bound the expected hitting time of the limit set, associate with each
  configuration~$c^t$ a weight~$X_t$ which is the maximum distance between two
  consecutive particles
  if there are at least two particles, or $N$ if there is only one
  particle, or $N+1$ if there is no particle. For all $t$, $X_t \in
  \{1,\ldots,N+1\}$ and $c_t$ belongs to the limit set if and only if $X_t\in\{N,N+1\}$. Let $\Delta X_{t+1}=X_{t+1}-X_t$. One can check that
  $\Esachant{\Delta X_{t+1}}{c_t=c} \geq 0$ for any configuration~$c$. Moreover
  $\Esachant{(\Delta X_{t+1})^2}{c_t=c} \geq 3/N$ for any
  configuration~$c$ not in the limit set. Consequently $X_t^2-3t/N$ is a sub-martingale and we can apply the Stopping Time Theorem 
  to the stopping time $T=\min\set{t \geq 0}{X_t \in \{N,N+1\}}$. It gives $\Esperance{T} = O(N^3)$ which is thus an upper bound on
  the expected hitting time of the limit set.

  This proof applies with no modification to the hitting time on graphs which are paths.
\qed
\end{proof}

\subsection{Long Range Interaction and Exponential Convergence (Trees)}
In this part, we introduce biased trees (Definition \ref{def:biased:tree} and
Figure \ref{fig:biased}) such that the dynamics $\cdual$ converges in
exponential time on this topology (Theorem \ref{the:biased:tree}). Vertices of
biased trees have degree at most~4. 
In fact, biased trees simulate biased random walks (Definition
\ref{def:biased:walk}) which converge in exponential time. Biased trees are
created from small trees called widgets (Definition~\ref{def:widget} and Figure
\ref{fig:exponential}) arranged on a line. Except at the ends, this line of
widgets is made of ``gates''. According to the configuration, these gates are
either locked, unlocked or stable (Definition~\ref{def:conf:gate}). On a correct
configuration (Definition \ref{def:correct}), the line of gates is split into
two regions: all gates on the left side are stable and all gates on the right
side are unstable (locked or unlocked). In a correct configuration, three
different events may be triggered with the same probability~$1/N$
(Fact~\ref{fact:biased:actice} and Corollary~\ref{cor:update}):

\begin{itemize}
\item the rightmost stable gate becomes an unlocked gate;
\item the leftmost unstable gate becomes stable if it is unlocked;
\item the leftmost unstable gate is switched from locked to unlocked or the contrary.
\end{itemize}

Thus stable gates tend to disappear. This dynamics will ultimately converge to
the stable configuration $\cdual_f$ (Definition
\ref{def:biased:final:configuration}). To reach this configuration all gates
must be stable. Thus it takes an exponential time for the dynamics $\cdual$ to converge
on a biased tree with an initial correct configuration.

\begin{definition}[Biased random walks]
\label{def:biased:walk}
A Biased Random Walk is a sequence of random variables $(X_i)_{i\geq 0}$ defined on $\{0,\ldots ,n\}$ such that for all~$i \geq 0$:
\begin{itemize}
\item $\Psachant{X_{i+1}=1 }{X_i=0}=1$ (reflecting barrier at~0).
\item $\Psachant{X_{i+1}=n}{X_i=n}=1$ (absorbing barrier at~$n$).
\item $\exists a,b \in \mathbb{R}_+ \quad\forall x\in\{1, ..., n-1\}\quad 
  \begin{cases}
    \Psachant{X_{i+1}=x-1}{X_i=x}\\
      \quad+\Psachant{X_{i+1}=x+1}{X_i=x}=1\\
    0 < a < \Psachant{X_{i+1}=x+1}{X_i=x} < b < 1/2
  \end{cases}$
\end{itemize}
\end{definition}

\begin{theorem}
Let $T:= \min \set{i\ge0}{X_i = n}$ be the absorption time at~$n$ and
for all $0 \leq k \leq n$, let $\EsperanceFrom{k}{T}:=\Esachant{T}{X_0=k}$ be its expectation starting from~$k$.
Then $$\theta_k(b) \leq \EsperanceFrom{k}{T} \leq \theta_k(a)$$
where $\theta_k(p)=\frac{2p(1-p)}{(1-2p)^2}\left(\left(\frac{1-p}{p}\right)^n-\left(\frac{1-p}{p}\right)^k\right)-\frac{n-k}{1-2p}$.
\end{theorem}
This theorem is a direct consequence of classical analysis of random
walks on~$\{0,\ldots,n\}$~where
\begin{itemize}
\item $\Psachant{X_{i+1}=1}{X_i=0}=1$
\item $\Psachant{X_{i+1}=n}{X_i=n}=1$
\item $\forall x \in \{1,...,n-1\}\quad
\begin{cases}
  \Psachant{X_{i+1}=x+1}{X_i=x}=p\\
  \Psachant{X_{i+1}=x-1}{X_i=x}=q
\end{cases}
$, with $p+q=1$.
\end{itemize}
Solving the following system of equations~\cite{GrimmettStirzaker2001}:
$$\begin{cases}
\EsperanceFrom{k}{T}=p(1+\EsperanceFrom{k+1}{T})+(1-p)(1+\EsperanceFrom{k-1}{T}) & \forall k \in\{0,...,n\}\\
\EsperanceFrom{n}{T}=0 & \\
\EsperanceFrom{0}{T}=1+ \EsperanceFrom{1}{T} 
\end{cases}
$$
one gets 
$\EsperanceFrom{k}{T}=\frac{2p(1-p)}{(1-2p)^2}\left(\left(\frac{1-p}{p}\right)^n-\left(\frac{1-p}{p}\right)^k\right)-\frac{n-k}{1-2p}$.

\begin{definition}[Widgets]
\label{def:widget}
A \emph{Widget} $W$ is a tree $\Trees = (\Cells, \Edges, b)$ where $b \in
\Cells$ is called the \emph{bridge}. We consider the three widgets described in
Figure $\ref{fig:exponential}$.a: \emph{head}, \emph{gate} and \emph{tail} and
the three configurations $\cdual_l, \cdual_u, \cdual_s$ for gates.
\end{definition}

\begin{figure}[htb]
\centerline{\scriptsize
\begin{tabular}{p{6cm}@{\qquad}p{6cm}}
   \,\! \hfill  \raisebox{.75cm}{\includegraphics[width=6cm]{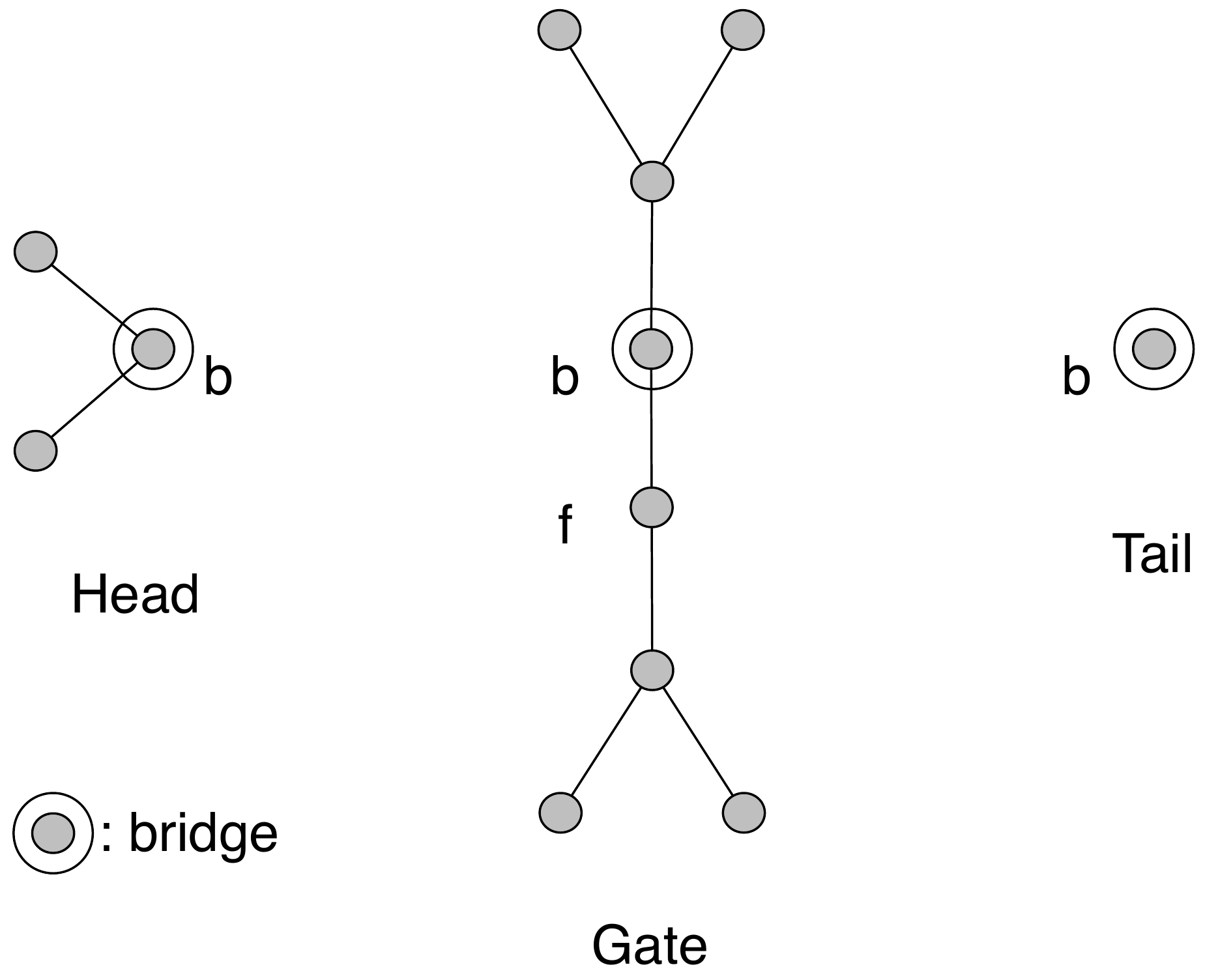}} \hfill ~
&   \,\! ~\hfill  \raisebox{0.75cm}{\includegraphics[width=6cm]{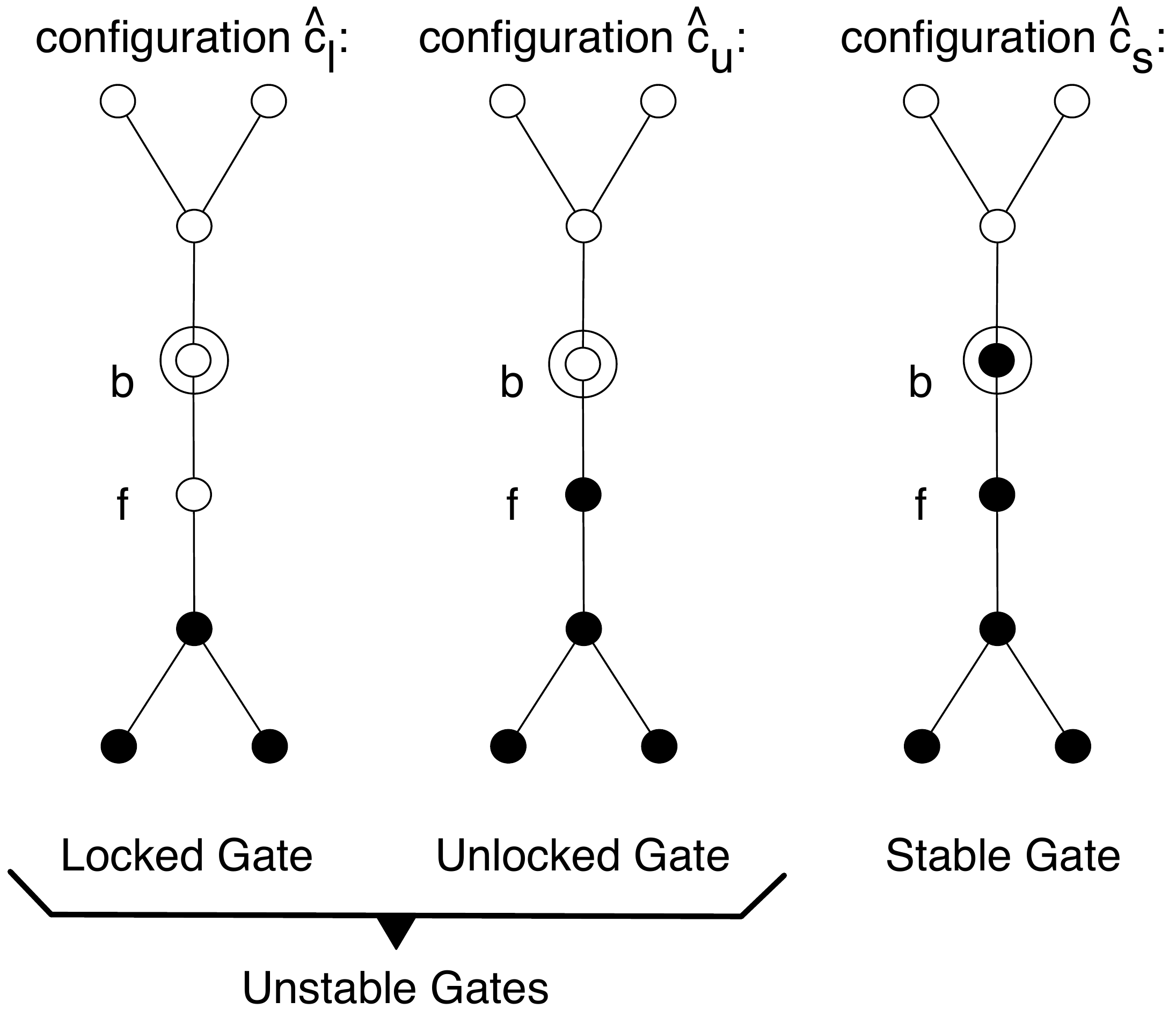}} \hfill  ~
\\[-6mm]
\multicolumn{1}{p{5.5cm}}{\small \ref{fig:exponential}.a -- The 3 widgets used in the construction of a biased tree.
  Gray denotes the fact that the vertex state is not represented.}
&
\multicolumn{1}{p{6cm}}{\small \ref{fig:exponential}.b -- The three configurations $\cdual_l$,$\cdual_u$ and $\cdual_s$.}
\end{tabular}}
\caption{Widgets used in the construction of a biased tree.}
\label{fig:exponential}
\end{figure}

\begin{definition}[Biased trees]
\label{def:biased:tree}
Let $(W_{i})_{0\leq i \leq n+1}$ be a finite sequence of widgets where $W_i =
(\Cells_i, \Edges_i,b_i)$. From this sequence, we define the
tree~$\Trees=(\Cells,\Edges)$ where~$\Cells = \cup_{i=0}^{n+1} \Cells_i$ and
$\Edges= (\cup_{i=0}^{n+1} \Edges_i) \bigcup (\cup_{i=0}^{n}
b_ib_{i+1})$. Abusively we also denote by $(W_{i})_{0\leq i \leq n+1}$ the tree
generated by this sequence. A \emph{biased tree} of size $n$ is a finite
sequence of widgets~$(W_{i})_{0\leq i \leq n+1}$ where $W_0$ is a head, for
$1\leq i \leq n$,~$W_i$ is a gate and $W_{n+1}$ is a tail.
\end{definition}

\begin{figure}[htb]
\includegraphics[width=12cm]{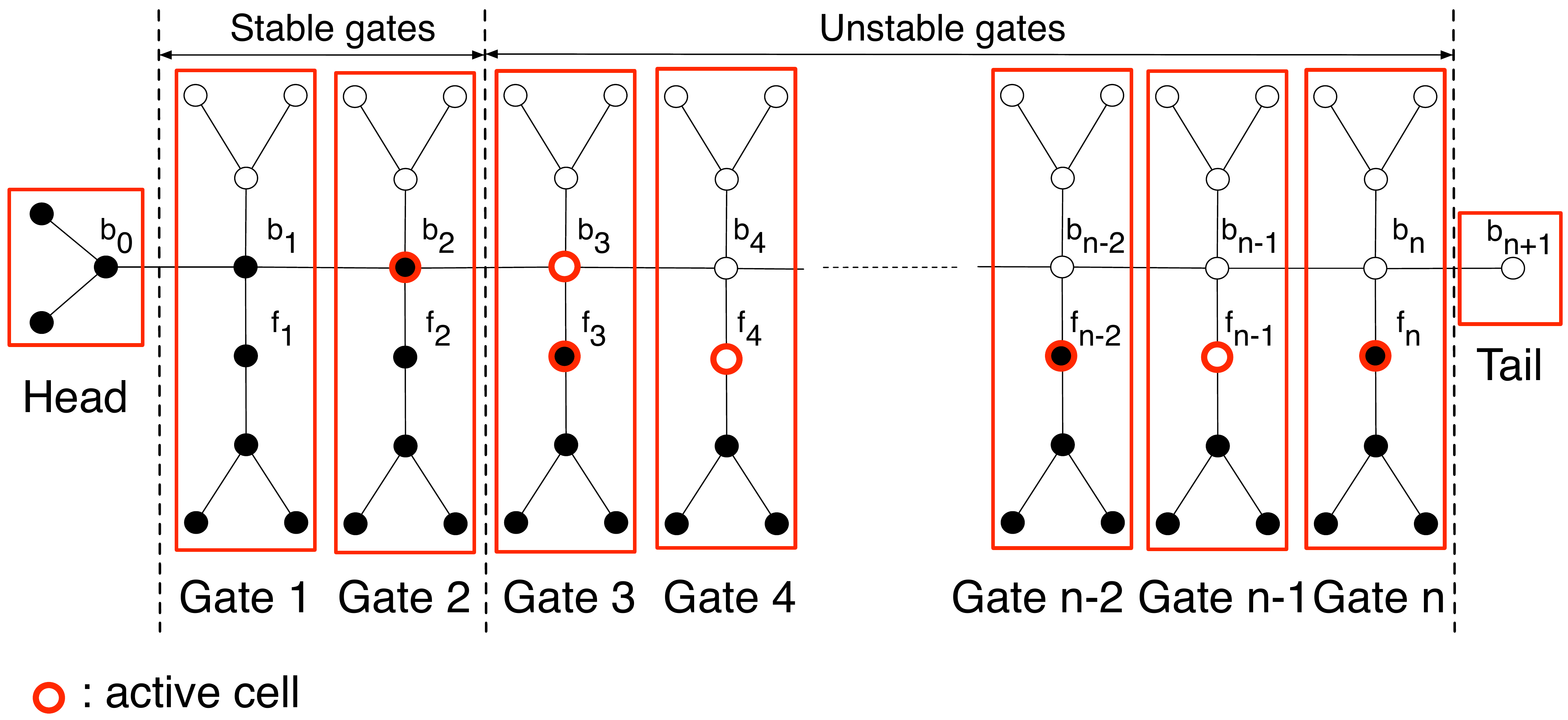}
\caption{A biased tree and a correct configuration on unlocked position $2$.}
\label{fig:biased}
\end{figure}

\begin{definition}[Stable and unstable gates]
\label{def:conf:gate}
Consider a biased tree~$(W_{i})_{0\leq i \leq n+1}$ and a configuration
$\cdual$. We denote by~$\cdual_{W_i}$, the restriction of~$\cdual$ to
widget~$W_i$. We say that gate $i$ is locked if $\cdual_{W_i} = \cdual_l$,
unlocked if $\cdual_{W_i} = \cdual_u$ and stable if $\cdual_{W_i} =
\cdual_s$. An unstable gate is a gate which is locked or unlocked.
\end{definition}

\begin{definition}[Correct configuration]
\label{def:correct}
Configuration $\cdual$ is \emph{correct} if vertices of the head are black, the
tail is white, and there exists a $j$ such that for all $1 \leq i \leq j$ gate $i$
is stable and for all $j < k \leq n$ gate $k$ is unstable. We say that
configuration $\cdual$ is on \emph{position} $j$. We denote by $\Pos(\cdual)$
the position of configuration $\cdual$. The position is unlocked if $j=n$ or
gate $j+1$ is unlocked, the position is locked otherwise.
\end{definition}

\begin{fact}[Active vertices]
\label{fact:biased:actice}
Consider a correct configuration $\cdual$ on position $j$. The active vertices of~$\cdual$ are:
\begin{itemize}
\item Vertex $b_j$ if $j \neq 0$.
\item Vertex $b_{j+1}$ if $j \neq n$ and gate $W_{j+1}$ is unlocked.
\item Vertex $f_i$ if $j < i \leq n$.
\item Vertex $b_{n+1}$ if $j = n$.
\end{itemize}
\end{fact}

\begin{proof}
  Consider a correct configuration $\cdual$ on position $j$. The only vertices
  which may be active in~$\cdual$ are vertices $b_i$ and $f_i$ for $ 1\leq i \leq
  n$ and vertex $b_{n+1}$. Vertex $b_{n+1}$ is active if and only if~$\cdual(b_{n})
  = 1$ that is to say gate $W_{n}$ is stable. For all $1\leq i \leq n $, vertex
  $f_i$ is active if and only if the gate $W_i$ is unstable that is to
  say~$j<i\leq n$.  For all $1\leq i \leq n$, vertex $b_i$ is inactive
  if~$\cdual(b_{i-1})= \cdual(b_{i}) = \cdual(b_{i+1})$. Thus among
  vertices~$(b_i)_{1 \leq i \leq n}$, only vertices $b_j$ and $b_{j+1}$ may be active:
  vertex $b_j$ is active and vertex $b_{j+1}$ is active if gate $W_{j+1}$ is
  unlocked.
\qed
\end{proof}

\begin{definition}[Final configuration]
\label{def:biased:final:configuration}
The \emph{final configuration} $\cdual_f$ is the configurations where vertices of
the head are black, the tail is black and every gate is stable. We say that
$\cdual_f$ is on position $n+1$, $\Pos(\cdual_f) = n+1$.
\end{definition}

\begin{lemma}
Configuration $\cdual_f$ is stable.
\end{lemma}

\begin{proof}
  Consider the correct configuration $\cdual$ on position $n$. According to
  fact~\ref {fact:biased:actice}, only vertices $b_{n}$ and $b_{n+1}$ are
  active. If vertex $b_{n+1}$ fires, these two vertices become inactive and no other
  vertex becomes active. Firing vertex $b_{n+1}$ leads to configuration
  $\cdual_f$. Thus~$\cdual_f$ is stable.
\qed
\end{proof}

\begin{corollary}
\label{cor:update}
Consider a correct configuration $\cdual$, then configuration $\cdual' =
\ddual(\cdual)$ is either correct or $\cdual_f$. Moreover
${|\Pos(\cdual')-\Pos(\cdual)|\leq 1}$.
\end{corollary}

\begin{proof}
  Consider a correct configuration $\cdual$ on position $j$ and the
  configuration $\cdual' = \ddual(\cdual)$. If an inactive vertex fires then
  $\cdual' = \cdual$. Now consider that an active vertex fires (see fact
  \ref{fact:biased:actice}):
\begin{itemize}
\item if $j\neq 0$ and vertex $b_j$ fires: then gate $W_j$ becomes unlocked and
  $\cdual'$ is a correct configuration on unlocked position $j-1$.
\item if $j\neq n$ and vertex $b_{j+1}$ fires: then gate $W_{j+1}$ becomes stable
  and $\cdual'$ is a correct configuration on position $j+1$.
\item if vertex $f_{i}$ fires with $j < i \leq n$: then gate $W_{i}$ becomes
  unlocked (resp. locked) in $\cdual'$ if it is locked (resp. unlocked) in
  $\cdual$. Configuration $\cdual'$ stays correct and on position $j$.
\item if $j=n$ and vertex $b_{n+1}$ fires: then $\cdual' =\cdual_f$.
\qed
\end{itemize}
\end{proof}

\begin{theorem}
\label{the:biased:tree}
On biased trees of size~$n$ (i.e. $N=8n+4$ vertices), starting from a correct
configuration, Stochastic Minority converges almost surely to $c^f$. Moreover
the hitting time~$T$ of the limit set satisfies ${\Theta(1.5^n)\leq
  \Esperance{T} \leq \Theta(n4^n)}$.
\end{theorem}
\begin{proof}
  Consider a biased tree of size $n$, an initial correct configuration
  $\cdual^0$ on position $0$ and the sequence $(\cdual^t)_{t \geq 0}$. Dynamics
  $\ddual$ converges almost surely from initial configuration $\cdual^0$
  and~$c^{T}=c_f$. We define the sequence of random variable~${(t_i)_{i\geq 0}}$
  as~$t_0 =0$ and
  $t_{i+1}=\min\set{t>t_{i}}{\Pos(\cdual^{t_{i+1}}) \neq \Pos(\cdual^{t_i})\text{ or }\Pos(\cdual^{t_{i+1}}) = n+1}$.
  Consider the sequence
  of random variable $(X_i)_{i \geq 0}$ such that $X_i =
  \Pos(\cdual^{t_i})$. According to corollary \ref{cor:update}, $|X_{i-1}-X_i| =
  1$. 

  Consider a configuration $\cdual^t$ on locked position $n > j > 0$ then firing
  vertex $b_j$ leads to a configuration on position $j-1$ and firing vertex $f_j$
  leads to a configuration on unlocked position $j$. Firing other vertices does not
  affect the position of the configuration. Consider a configuration $\cdual^t$
  on unlocked position $n > j > 0$ then firing vertex $b_j$ leads to a
  configuration on position $j-1$, firing vertex $f_j$ leads to a configuration on
  locked position $j$ and firing vertex $b_{j+1}$ leads to a configuration on
  position $j+1$. Firing other vertices does not affect the position of the
  configuration. A vertex has a probability $1/N$ to fire where $N=4+8n$. Thus,
  the evolution of a configuration on position $0 < j < n$ can be summarized as:

\centerline{\includegraphics[width=\linewidth]{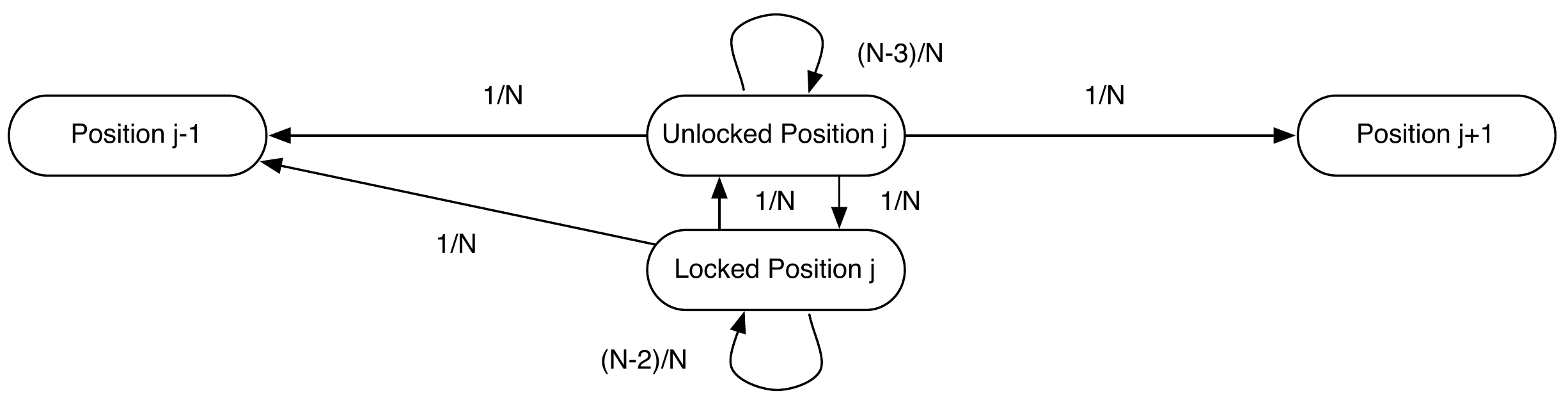}}

A basic analysis yields that:

\begin{itemize}
\item if $1 \leq x  \leq n$ then $\Psachant{X_{i+1}=x+1}{X_i=x} = 1-\Psachant{X_{i+1}=x-1}{X_i=x}$
  and $1/5 \leq \Psachant{x_{i+1}=x+1}{X_i=x} \leq 2/5$.
\item $\Psachant{X_{i+1}=1}{X_i=0}=1$.
\item $\Psachant{X_{i+1}=n+1}{X_i=n+1}=1$.
\end{itemize}

Thus the behavior of $(X_i)_{i\geq0}$ is as described in definition
\ref{def:biased:walk}.  We define the random variable $T'=\min\set{i\ge0}{X_i=n}$ which
corresponds to the first time when all gates are stable, then
$\Theta(\frac{3}{2}^n) \leq \Esperance{T'} \leq \Theta(4^n)$ 
(see
Def. \ref{def:biased:walk}). We call $c^{f-1}$ the correct configuration on
position $n$ (\emph{i.e.} all gates are stable). Then $c^T=c^f$, $c^{T-1} =
c^{f-1}$ and $\Psachant{c^{t+1}=c^f}{c^t=c^{f-1}} =1/2$. Thus, $\Esperance{T} = \Theta
(\Esperance{t_{T'}}$). By definition, $t_{T'}= \sum_{i=1}^{T'}(t_{i}-t_{i-1}) =
\sum_{i=1}^{\infty}[(t_{i}-t_{i-1})1_{t_i<T'}]$. Since there are at most $2$
vertices which may modify the position of a correct configuration, we have $1 \leq
\Esperance{t_{i+1}-t_{i-1}} \leq \Theta (n)$. Thus $\sum_{i=1}^{\infty}(1_{t_i<T'}) \leq
\Esperance{t_{T'}} \leq \Theta (\sum_{i=1}^{\infty}(n1_{t_i<T'}))$. We conclude that
${\Theta(\left(\frac{3}{2}\right)^n)\leq \Esperance{T} \leq \Theta(n4^n)}$.
\qed
\end{proof}

\subsubsection{Subcase: Binary Trees Converge Rapidly}
In this section, we note that
on binary trees, i.e. trees where the degrees are at most~3, the dynamics ends by fixing the states
of the vertices of degree~1 and~3 (see Remark~\ref{rem:odd-degree})
and some isolated particles may remain and oscillate on disjoint paths.

\begin{definition}[Path]
  \parpic[r]{
    \scalebox{0.6}{
\begin{tikzpicture}[grow cyclic,shape=circle]
\tikzstyle{level 1}=[level distance=10mm,sibling angle=120]
\tikzstyle{level 2}=[level distance=10mm,sibling angle=120]
\tikzstyle{level 3}=[level distance=10mm,sibling angle=120]
\tikzstyle{level 4}=[level distance=10mm,sibling angle=120]
\tikzstyle{every node}=[fill=black!20!white,draw=black,shape=circle,cap=round]
\coordinate[rotate=-90] 
node[fill=black] {}
child {node {}}
child {
  node[fill=black] {}
  child {
    node[fill=black]{}
    child{
    node[fill=black]{}
    child{
    node[fill=black]{}
}}}}
child {node {}}
;

\end{tikzpicture}
}
}
In this subcase, we call \emph{path} a connected subgraph
where all nodes but the end nodes have degree 2
(end nodes must thus have degree 1 or 3). An example is composed of the black nodes on this figure:
\end{definition}

\begin{theorem}
\label{thm:tree-deg3}
Stochastic Minority on trees with degrees at most~3 hits the limit set
in~$O(N^4)$ steps on expectation. The attractors of a tree~$\Trees$ are in
bijection with the matchings of the reduced tree~$\Trees'$ where each path
of~$\Trees$ has been replaced by an edge, then each leaf has been removed.
\end{theorem}
\begin{proof}
  We study here the movements of particles in the initial tree~$\Trees$.
  One can divide~$\Trees$ into its induced subgraphs which are paths.
  Those paths link the vertices of odd degree (1 or
  3). The reduced tree~$\Trees'$ is obtained by replacing each path 
  by an edge, then removing the leaves.

  Consider a configuration~$c$ on~$\Trees$ which belongs to an
  attractor. There cannot be two particles on the same path
  , otherwise a sequence of updates could lead to the collision of these two
  particles and thus to an energy decrease. In the same way, there cannot be two
  particles on two
  paths 
  which share a common extremity $i$.
  This extremity would necessarily be a vertex of degree~3, then a
  sequence of updates could position the two particles on the edges incident to~$i$.
  Firing $i$ at that time would decrease the number of
  particles by at least~1 and lead to an energy decrease.
  Finally, there cannot be a particle on a path which has an end of degree one,
  since the particle could disappear at this end.

  \parpic[r]{\scalebox{0.6}{\begin{tikzpicture}[grow cyclic,shape=circle]
\tikzstyle{level 1}=[level distance=15mm,sibling angle=120]
\tikzstyle{level 2}=[level distance=12mm,sibling angle=120]
\tikzstyle{level 3}=[level distance=10mm,sibling angle=120]
\tikzstyle{every node}=[fill=black!20!white,draw=black,shape=circle,cap=round]
\coordinate[]
node{}
child foreach \x in {1,2,3} {
  node {}
  child foreach \y in {1,2} {
    node{}
    child foreach \z in {1,2} {
    node{}
  }
  \particule12
  \particule21
}
\ifnum\x=3
  edge from parent[]
  node[color=red,left=2mm,pos=0.3]{}
  \fi
}
;
\end{tikzpicture}
}}
  Reciprocally, for configurations
  where there are no particles on the same path nor on adjacent paths
  nor on a path with an end of degree 1,
  the number of particles on each path is constant.
  This proves that the energy cannot decrease, and that
  such configurations belong to the limit set.
  This also establishes the bijection between attractors
  and matchings of~$\Trees'$ (the matching indicates where the isolated particles are located in~$\Trees$).
  An illustration of $\Trees'$ is on the opposite figure, particles are next to the edges.

  \medskip
  To prove the bound on the expected hitting time of the limit set, we find a
  bound on the time until at least one particle disappears. Consider a
  configuration where there exist two particles on a same path 
  of length~$n$. One can suppose that these two particles follow a random walk
  on this path with reflecting barriers at each extremity, unless they collide
  with another particle (leading to the loss of two particles in the tree) or
  unless one of the two particles leaves the path (leading to the loss of at
  least one particle in the tree since leaving necessarily involves a vertex of
  degree~3 fired with two incident particles). Thus under the condition that
  they have not disappeared before, a bound on the expected time elapsed until
  they collide can be derived from classical studies of random walks with
  reflecting barriers~\cite{GrimmettStirzaker2001}: this expected time is
  bounded by~$O(n^3)$, and consequently it is also a bound on the time until at
  least one particle disappear.

  \parpic[r]{\scalebox{0.6}{
\begin{tikzpicture}[grow cyclic,shape=circle]
\tikzstyle{level 1}=[level distance=10mm,sibling angle=120]
\tikzstyle{level 2}=[level distance=10mm,sibling angle=120]
\tikzstyle{level 3}=[level distance=10mm,sibling angle=120]
\tikzstyle{level 4}=[level distance=10mm,sibling angle=120]
\tikzstyle{every node}=[fill=black!20!white,draw=black,shape=circle,cap=round]
\coordinate[] 
node (i) {}
child foreach \x in {1,3,2} {
  node {}
  child {
    node{}
    child{
    node{}
    child foreach \y in {1,2} {
      node{}
    }}
  \ifnum\x<3
  edge from parent[]
  node[color=red,left=2mm,pos=0.2]{}
  \fi
}
}
;
\node[fill=none, draw=none, above right of=i, node distance=4mm] {$i$};

\end{tikzpicture}
}}
  Consider a configuration where there exist two particles on two
  paths 
  of respective lengths~$n$ and~$m$, sharing a
  common extremity: the vertex~$i$.
  An illustration is on the opposite figure.
 With the same reasoning, assuming that the two
  particles have not led to the removal of another particle means that they
  follow random walks on their respective paths with reflecting barriers at the
  extremities. Then they can only disappear by being both incident to vertex~$i$
  when vertex~$i$ is fired. By analyzing the two-dimensional random walk
  corresponding to the evolution of the respective distances to vertex~$i$, this
  event occurs after at most~$O(\max(n,m)^3)$ steps on expectation, as proved in
  Section~\ref{hitting-time-rectangle} (or in~\cite{minority-on-trees_techreport})
  using standard tools for multi-dimensional random walks.

  Finally, from any configuration which does not belong to the limit set, at
  least one particle disappear within $O(N^3)$ steps on expectation. Since the
  number of particles in any initial configuration on a tree is bounded by~$N$,
  the expected time to hit the limit set is bounded by~$O(N^4)$.
\qed
\end{proof}

\subsection{Phase transition}
\label{sec/phase-transition}

In this part, we consider the infinite graph 
where 
the set of vertices is $\N$
and vertex $i$ has two neighbors: $i-1$ and $i+1$. We consider the initial configuration $c^{init}$ where the state of vertex~$i$ is $1$ if $i=0$ and $i\!\!\mod 2$ otherwise.
This configuration possesses only three vertices which are not in their minority state: $-1$, $0$ and $1$. Updating vertex $0$ leads to a stable configuration.

We consider the fully asynchronous dynamics: at each time step, only one vertex is updated and this vertex is selected uniformly at random among the active vertices. Note that the set of active vertices is always finite, so this random selection is feasible.
The limit set of an execution of Minority starting from the initial configuration $c^{init}$ is composed of the stable configuration where the state of vertex $i$ is $i \mod 2$.

We denote by $\ProbaFrom{\alpha}{c^{init}}$ the probability that the dynamics reaches a stable configuration from the initial configuration $c^{init}$. If $\ProbaFrom{\alpha}{c^{init}}<1$ then the expected time to reach the limit set is infinite. Experimentally, there is a phase transition on $\ProbaFrom{\alpha}{c^{init}}$ depending on $\alpha$ with a critical value $\alpha_c \approx 0.5$ such that:
\begin{itemize}
\item if $\alpha<\alpha_c$ then $\ProbaFrom{\alpha}{c^{init}}=1$;
\item if $\alpha>\alpha_c$ then $\ProbaFrom{\alpha}{c^{init}}<1$
\end{itemize} 
Our result link the critical value $\alpha_c$ to the critical value $0.6298<p_c<2/3$ of directed percolation.

\begin{theorem}
\label{th/percolation}
If $\alpha \geq \sqrt[3] {1-(1-p_c)^4}$ then $\ProbaFrom{\alpha}{c^{init}} < 1$. The quantity $\sqrt[3] {1-(1-p_c)^4}$ is in $[0.993;0.996]$; thus if $\alpha \geq 0.996$ then $\ProbaFrom{\alpha}{c^{init}} < 1$.
\end{theorem}

\begin{proof}
In~\cite{R08perco}, this theorem is proved for this dynamics with the rule \textsc{flip-if-not-equal} (an updated vertex switches its state if at least one of its neighbor is in a different state than itself) from the initial configuration $c'^{init}$ where vertex $i$ is in state $0$ if $i\neq 0$ and $1$ otherwise. We simply prove here that \textsc{flip-if-not-equal} is the dual dynamics of Minority obtained by switching state of vertices $i$ if $i\mod 2=1$:
\begin{itemize}
\item if vertex $i$ is active for Minority then at least one if its neighbors is in the same state as itself. Thus in the dual configuration, vertex $i$ has at least one of its neighbors in a different state as itself.
\item if vertex $i$ is inactive for Minority then both its neighbors are in a different state than itself. Thus, in the dual configuration, vertex $i$ has both its neighbors in the same state as itself.
\end{itemize} 
Also, $c'^{init}$ is the dual configuration of $c^{init}$.
\end{proof}

\section{Conclusion}
The table below sums up the different worst case average hitting times of 
a limit set for different topologies and compares the fully
asynchronous dynamics to the synchronous one. In the case of
the torus under fully asynchronous dynamics, it is conjectured
that this average ``convergence'' time admits a polynomial bound
in the number of cells.  

\centerline{\small \begin{tabular}{|c|c|c|c|}
\hline 
& Fully Asynchronous & Synchronous
\\ \hline
Path or Cycle & Poly & Exp~\cite{GM90}
\\ \hline
Tree, max degree $\leq 3$ & Poly & Exp~\cite{GM90}
\\ \hline
Tree, max degree $\geq 4$ & Exp & Exp~\cite{GM90}
\\ \hline
Torus, von Neumann neighborhood & Poly ?~\cite{RST-TCS2009} & Exp~\cite{GM90}
\\ \hline
Torus, Moore neighborhood & Poly ?~\cite{RST-DMTCS2010} & Exp~\cite{GM90}
\\ \hline
Clique & Poly & Poly~\cite{GM90}
\\ \hline
\end{tabular}}

The Minority rule admits a rich range of behaviors under full asynchronism.
The case of trees has shown that the average hitting time of limit sets
is not necessarily polynomial under full asynchronism (there is a threshold on the maximum degree). For now, it is an open
problem to predict from the graph topology whether the dynamics will converge
fast or slowly. Following this work, a challenge is to identify the graph 
parameters that guide this convergence speed, as well as extending such results
to other updating rules.

\subsubsection*{Acknowledgements.}
We are very grateful to the referees for their comments and suggestions which
have improved the paper.

\bibliographystyle{splncs}
\bibliography{biblio}

\appendix

\section{Characterization of stable configurations on trees}
\label{sec/char-stable-conf}
Let us now study more precisely the structure of the limit set.
This characterization allows to 
count the number of attractors on a given tree.

\bigskip

We first pick an arbitrary vertex of degree 1 and set it as root~$r$ in the tree
(this introduces ``parent'' and ``children'' relations).
The term ``degree'' refers to the original graph:
a vertex of degree three has two children.

We assign a label to each
vertex to count the number of attractors and the size of the limit set.
The number of attractors will be shown to be the number of acceptable labelings.
The set of labels we use is $\{(\fixed,0), (\fixed,1), (\oscil,0), (\oscil,1)\}$.
``$\oscil$'' intuitively means ``oscillating''
while
``$\fixed$'' means ``fixed''
(like the recorder symbols play/stop).
The second component is called the ``preferred'' state of the vertex.

\begin{definition}[acceptable labeling]
  A labeling is acceptable if and only if, for each vertex $i$, if the vertex has label
  \begin{enumerate}
  \item \label{def:acceptable-label:enum:fixed}
    $(\fixed,\alpha)$ then it has strictly more than $\deg(i)/2$ neighbors with label $(\fixed,\alpha)$;
  \item \label{def:acceptable-label:enum:oscil}
    $(\oscil,\alpha)$ then
    \begin{enumerate}
    \item \label{def:acceptable-label:enum:oscil:root}
      if the parent has a label of the form $(\fixed,\beta)$,
      then $\alpha=0$ and
       $i$ has one more child labeled $(\cdot,1-\beta)$ than children labeled $(\cdot,\beta)$;
    \item \label{def:acceptable-label:enum:oscil:nonroot}
      otherwise, i.e. the parent has a label of the form $(\oscil,\cdot)$,
      $i$ has one more child labeled $(\cdot,\alpha)$ than children labeled $(\cdot,1-\alpha)$.
    \end{enumerate}
  \end{enumerate}
\end{definition}
\noindent
Note that only vertices of even degree can have a label of the form $(\oscil,\cdot)$.
The apparent asymmetry in case~\ref{def:acceptable-label:enum:oscil:root}
(imposing $\alpha=0$) is there only to avoid double counts in
Theorem~\ref{nb-attr=nb-label}, one could as well have defined acceptable
labelings with $\alpha=1$.

Theorem~\ref{nb-attr=nb-label} shows that a labeling corresponds to an
attractor, and Theorem~\ref{attractor-structure} details the meaning of a
labeling, thus the structure of an attractor.

For a labeling $L$, we define a few special configurations:
$\snd(L)$ is the projection of the second component: $L_i=(\cdot,\alpha) \iff (\snd(L))_i=\alpha$,
and $\paint(L,\alpha)$ sets all vertices labeled $(\oscil,\cdot)$ to $\alpha$:
$$(\paint(L,\alpha))_i:=
\begin{cases}
  \beta&\text{if } L_i=(\fixed,\beta)\\
  \alpha&\text{if } L_i=(\oscil,\cdot)
\end{cases}
$$
For a configuration $c$ in the limit set, we denote $\attr(c)$ the attractor containing~$c$.

\begin{lemma}
  \label{prop:snd-and-paint-in-limit-set}
 If $L$ is an acceptable labeling, 
 $\snd(L)$ and $\paint(L,\alpha)$ are in the limit set.
\end{lemma}
\begin{proof}
  Consider the configuration $\paint(L,0)$.
  We consider the vertices in a bottom-up order (a vertex is considered after all its children),
  and update each vertex which is not in its preferred state.
  Updating a vertex makes it goes to its preferred state thanks to the definition of acceptable labelings.
  The current configuration is $\snd(L)$.

  Then, we consider the vertices in a top-down order and update those that are in state $0$.
  When a vertex is considered, its children are in their preferred state,
  and if the parent of a vertex is labeled $(\oscil,\cdot)$ then this parent is in state $1$.
  Again, thanks to the definition of acceptable labelings, updating a vertex makes it go to state~$1$.
  We get the configuration $\paint(L,1)$.

  By symmetry, there is a sequence of updates from $\paint(L,1)$ to $\paint(L,0)$.
  Moreover, there is no sequence of updates leading from one of this configurations
  to change of state of a vertex labeled $(\fixed,\alpha)$:
  the first change of state of such a vertex $i$ would contradict the fact that it has more
  than $\deg(i)/2$ neighbors with label $(\fixed,\alpha)$ and thus in the state~$\alpha$.
  Which implies that, whatever the configuration reached from $\paint(L,0)$,
  there is by monotonicity a sequence of updates yielding the configuration $\paint(L,1)$.
  The full cycle is in the limit set.
\qed
\end{proof}

\begin{theorem}
\label{nb-attr=nb-label}
Given a tree, there is a bijection between attractors and acceptable labelings.
\end{theorem}
\begin{proof}
  We define a function $f$ that maps an attractor $A$ to an acceptable labeling $L$.
  We then shows that for every attractor $A$, $\attr \circ \snd(f(A))=A$ and
  for every acceptable labeling $L$, $f(\attr\circ\snd(L))=L$.
  This imply that $f$ is a bijection.

\medskip
  To define $f$,
  let $A$ be an attractor, we construct $L=f(A)$.
  For all vertices that have the same state $\alpha$ in any configuration of $A$, define $L_i:=(\fixed,\alpha)$.
  All leaves are now labeled. We define the labeling of the remaining vertices inductively (in a bottom-up order).
  Each remaining vertex is oscillating, thus has an even degree and an odd number of children.
  Considering a vertex $i$ having all its children labeled:
  \begin{itemize}
  \item if the parent is already labeled (thus with a label of the form $(\fixed,\cdot)$),
    define $L_i:=(\oscil,0)$;
  \item otherwise, let $\alpha$ be the majority state among the preferred states of the children
    and define $L_i:=(\oscil,\alpha)$.
    (Since $i$ has an odd number of children, the majority state is well defined).
  \end{itemize}

\bigskip
Let us show that the labeling we have just defined is acceptable.
With the same argument as the proof of proposition~\ref{prop:algorithm-limit-set},
the configuration where all oscillating vertices are in state $\alpha$ is in $A$.
This configuration is $\paint(L,\alpha)$.

\begin{enumerate}
\item
  Consider a vertex $i$ labeled $(\fixed,0)$ and the configuration $\paint(L,1)$.
  All the neighbors of $i$ not labeled $(\fixed,0)$ are in the state $1$.
  So, there are necessarily more than $\deg(i)/2$ neighbors labeled $(\fixed,0)$ (and thus in the state $0$),
  otherwise, updating $i$ would make it change its state, contradicting the definition of $L_i=(\fixed,\cdot)$.
  By symmetry, point~\ref{def:acceptable-label:enum:fixed} of the definition of acceptable labelings is fulfilled.
\item We show point~\ref{def:acceptable-label:enum:oscil} inductively in a bottom-up order.
  Consider a vertex $i$ labeled $(\oscil,\alpha)$.

  Let us first show that there is a configuration $c\in A$ where all the vertices of the subtree having $i$ as root
  are in their preferred state, except $i$ (which is in state $1-\alpha$).
  Indeed, consider $\paint(L,1-\alpha)$ and update the vertices not in their preferred state in a bottom-up order.
  Thanks to the induction hypothesis, the label of those vertices is acceptable,
  thus updating them makes them change to their preferred state.
  This is a sequence of updates from a configuration in $A$ leading to $c$, so $c$ belongs to $A$.

  We can now show that the label of $i$ is acceptable:
  \begin{enumerate}
  \item If the parent has a label of the form $(\fixed,\cdot)$ then $L_i$ has been defined as $(\oscil,0)$.
    Moreover, in $c$, the parent of $i$ is in state $\beta$.
    In this configuration, updating $i$ cannot make the energy decrease,
    so $i$ must have one more child labeled $(\cdot,1-\beta)$ than children labeled $(\cdot,\beta)$.
    So, point~\ref{def:acceptable-label:enum:oscil:root} is fulfilled.
  \item Otherwise, the parent is labeled $(\cdot,\beta)$.
    From the definition of $L$, in configuration $c$, $i$ has a majority of children in state $\alpha$.
    Updating it thus makes it change its state.
    So $i$ must have as many neighbors in state $\alpha$ than in state $1-\alpha$.
    Which means that $\beta=1-\alpha$ and $i$ has exactly one more child labeled $(\cdot,\alpha)$ than children labeled $(\cdot,1-\alpha)$.
    Point~\ref{def:acceptable-label:enum:oscil:nonroot} is fulfilled.
  \end{enumerate}
\end{enumerate}

\bigskip
Let us now show that for any acceptable labeling $L'$, $f(\attr\circ\snd(L'))=L'$.
Since $\paint(L',0)$ and $\paint(L',1)$ are in the attractor $\attr\circ\snd(L')$
(cf. Lemma~\ref{prop:snd-and-paint-in-limit-set}),
the vertices labeled $(\oscil,\cdot)$ in $L'$ are oscillating in this attractor.
From the definition of $f$, these vertices are labeled $(\oscil,\cdot)$ in $f(\attr\circ\snd(L'))$.

Recall that there is no sequence of updates leading from $\snd(L')$
to change the state of a vertex labeled $(\fixed,\alpha)$:
the first change of state of such a vertex $i$ would contradict the fact that it has more
than $\deg(i)/2$ neighbors with label $(\fixed,\alpha)$ and thus in the state~$\alpha$.
This allows us to conclude that the labelings $L'$ and $f(\attr\circ\snd(L'))$
have the same cells labeled $(\fixed,0)$ and $(\fixed,1)$, thus they are equal.
(Indeed,
in the definition of $f$, the value of $\alpha$ for vertices labeled $(\oscil,\alpha)$
is entirely determined by the labeling of vertices labeled $(\fixed,\cdot)$.)

\bigskip
Finally, let $A'$ be an attractor, we show that $\attr\circ\snd(f(A'))=A'$.
$f(A')$ is an acceptable labeling so (Lemma~\ref{prop:snd-and-paint-in-limit-set}),
$\paint(f(A'),1)$ is in the attractor $\attr\circ\snd(f(A'))$.
Moreover, we have already noted that
the configuration where all the oscillating vertices of $A'$ are in state $1$
belongs to $A'$ (same argument as the proof of Proposition~\ref{prop:algorithm-limit-set}).
From the definition pf $f$, this configuration is $\paint(f(A'),1)$.
The attractors $\attr\circ\snd(f(A'))$ and $A'$
have the element $\paint(f(A'),1)$ in common, so they are the same attractor.

This concludes the proof by implying that $f$ is a bijection.
\qed
\end{proof}

\begin{theorem}[Structure of an attractor]
  \label{attractor-structure}
  Let $L$ be an acceptable labeling.
  Then for every configuration $c$ reachable by a sequence of updates from $\snd(L)$, for every vertex $i$:
  \begin{enumerate}
  \item \label{prop:structure:enum:fixed}
    If $L_i=(\fixed,\alpha)$ then $c_i=\alpha$ (this is why ``$\fixed$'' intuitively means ``fixed'').
  \item If $L_i=(\oscil,\alpha)$ (in this case $\deg(i)$ is even: $i$ is not the root, which has degree 1) then
    \begin{enumerate}
    \item\label{prop:structure:enum:oscil:root}
      if the parent has a label of the form $(\fixed,\cdot)$,
      $i$ is in the state appearing in majority among its neighbors (no constraint in case of equality);
    \item \label{prop:structure:enum:oscil:nonroot}
      otherwise
      \begin{itemize}
      \item if $i$ is in its preferred state $\alpha$,
        its children labeled $(\cdot,\alpha)$ are in their preferred state $\alpha$
      \item otherwise, all its children not in their preferred state are in the same state as $i$ (the state $1-\alpha$).
      \end{itemize}
    \end{enumerate}
  \end{enumerate}
\end{theorem}
\begin{proof}
  Recall that there is no sequence of updates from $\snd(L)$ to the firing of a vertex labeled $(\fixed,\alpha)$.
  This ensures point~\ref{prop:structure:enum:fixed}.

  Since $\snd(L)$, thus $c$, are in the limit set, the energy cannot decrease.
  It follows that any vertex $i$ has always at least $\deg(i)/2$ neighbors in the same state as $i$.
  This proves point~\ref{prop:structure:enum:oscil:root}.

  We prove the last point (\ref{prop:structure:enum:oscil:nonroot}) by recursion.
  Note that the assertion of the theorem is clearly true for $\snd(L)$.
  Let $c$ be a configuration verifying the assertion,
  $c'$ a configuration reached from $c$ by updating a vertex $i$,
  it is sufficient to prove that the assertion is true for $c'$.
  Precisely, we prove that it is true for $i$ and its neighbors.

  \noindent
  For $i$:
  \begin{itemize}
  \item If $i$ is in its preferred state $\alpha$ in $c$,
    then its children labeled $(\cdot,\alpha)$
    (there are $\deg(i)/2$ such children thanks to the definition of acceptable labelings)
    are then in state $\alpha$ in $c$ and thus in $c'$.
    Which means that $i$ can change its state only if all other children are in state $1-\alpha$.
    Point~\ref{prop:structure:enum:oscil:nonroot} stays true for $i$ in $c'$.
  \item Otherwise,
    it is either inactive and point~\ref{prop:structure:enum:oscil:nonroot} stays true in $c'$, or active.
    In the latter case, $i$ has as many neighbors in each state (because $c$ is in the limit state),
    all children of $i$ not in their preferred state are in the same state as $i$ (by recursion hypothesis)
    and $i$ has one more child labeled $(\cdot,\alpha)$ than children labeled $(\cdot,1-\alpha)$ (definition of acceptable labelings).
    These three conditions imply that all children of $i$ are actually in their preferred state.
    So, $c'$ fulfills point~\ref{prop:structure:enum:oscil:nonroot}.
  \end{itemize}

  For a neighbor $w$ of $i$:
  \begin{itemize}
  \item If $w$ is a child of $i$, nothing has changed for the sons of $w$:
    point~\ref{prop:structure:enum:oscil:nonroot} stays true.
  \item If $w$ is the parent of $i$, it is sufficient to consider the case where $i$ is in the same state $\alpha$ as $w$
    (there is no condition to check in the other case).
    \begin{itemize}
    \item If $\alpha$ is the preferred state of $i$ 
      then $i$ has at least $\deg(i)/2+1$ neighbors in state $\alpha$: its children labeled $(\cdot,\alpha)$ and its parent $w$. So  $i$ is inactive.
    \item Otherwise, $i$ and $w$ are not in their preferred state,
      and point~\ref{prop:structure:enum:oscil:nonroot} stays true, whether or not $i$ change its state.
      \qed
    \end{itemize}
  \end{itemize}
\end{proof}

\section{Omitted proofs: Bound on the Hitting Time of a 2D Finite Markov Chain}
\label{hitting-time-rectangle}
\subsection{Background on Markov chain theory}
\label{markov-chain-theory}
We recall only the necessary background on Markov chains to get a bound on the hitting time of a 2D finite Markov chain.
For a gentle introduction and proofs, we refer for instance to Chapter~10 of~\cite{LPW2006}.

Let $\famille Xt{\N}$ be a Markov chain.
We note $\tau_b$ the \concept{hitting time} of $b$, i.e. the first time the Markov chain is in state $b$:
$$\tau_b:=min\set {t\ge 0}{X_t=b}\ .$$
If $(X,P)$ is a Markov chain reversible with respect to the probability $\pi$,
the \concept{conductance} of an unoriented edge $(x,y)$ is
$$g(x,y):=\pi(x)P(x,y)=\pi(y)P(y,x)\ .$$
Let $a$ and $b$ be two distinguished vertices, representing source and sink.
We use the following potential, or \concept{voltage}, of a vertex:
$$V(x):=\Psachant{\tau_a<\tau_b}{X_0=x}\ .$$
Clearly $V(a)=1$ and $V(b)=0$.
Now, define the \concept{current flow} on oriented edges as
$$I(x,y):=g(x,y)\setp{V(x)-V(y)}\qquad \text{and} \qquad \norme{I}:=\sum_xI(a,x)\ .$$
The \concept{effective resistance} between $a$ and $b$ is
$${\cal R}(a,b):=\frac{V(a)-V(b)}{\norme I}\ .$$

\begin{theorem}[Commute time identity]
$$\Esachant{\tau_b}{X_0=a}+\Esachant{\tau_a}{X_0=b} = g{\cal R}(a,b)$$
\end{theorem}
\noindent 
In our case, $g=1$, so ${\cal R}(a,b)$ is an upper bound for the average hitting time $\Esachant{\tau_b}{X_0=a}$.
Here is how one can bound ${\cal R}(a,b)$.
A \concept{flow} from $a$ to $b$ is a function on oriented edges which is
antisymmetric: $\theta(x,y) = -\theta(y,x)$ and which obeys Kirchhoff's vertex
law:
$$\forall v\notin\{a,b\}\qquad \sum_x\theta(x,v)=0\ .$$
This is just the requirement ``flow in equals flow out''.
The \concept{strength} of a flow is
$$\norme{\theta}:=\sum_x\theta(a,x)\ .$$

\begin{theorem}[Thomson's Principle]
  For any finite connected graph,
  $${\cal R}(a,b) = \inf\set{{E}(\theta)}{\theta\text{ a unit flow from $a$ to $b$}}$$
  where $\displaystyle{E}(\theta):=\sum_{x,y}\frac{(\theta(x,y))^2}{g(x,y)}$.
\end{theorem}

\subsection{Application to our case}

If a neighbor does not exist (because the vertex is on the border), the edge points to the vertex itself.
Each edge has the same probability $\frac14$.

An invariant probability is the uniform probability $\pi:x\mapsto\frac1{nm}$.
Our Markov chain is reversible with respect to this probability.
So we can use the definitions of section~\ref{markov-chain-theory}:
$$\forall x,y\qquad g(x,y)=\frac1{4nm}\ .$$
Thanks to Thomson's principle, it is sufficient to construct a flow from $a$ to
$b$ to get an upper bound on ${\cal R}(a,b)$.
If $a=(i,j)$ and $b=(n,m)$, we consider the trivial (and far from optimal) flow
$$\begin{cases}
  \theta((k,j),\;(k+1,j)):=1 & \text{if } i\le k<n\\
  \theta((n,k),\;(n,k+1)):=1 & \text{if } j\le k<m\\
  \theta(x,y):=0 & \text{elsewhere.}\\
\end{cases}$$
That is, a flow on a single path from $a$ to $b$.

${E}(\theta)\le (n+m)4nm$. We conclude with the commute time identity that the average hitting time is $O(n^3)$,
assuming w.l.o.g. that $n\ge m$.

\end{document}